\documentclass[journal]{IEEEtran}
\usepackage{enumerate}
\usepackage{amssymb,stmaryrd,amsmath,amsfonts,rotating}
\usepackage[noadjust]{cite} \usepackage{color} 
\usepackage[vflt]{floatflt} \usepackage{epic}
\usepackage{color}
\usepackage{verse}
\usepackage[mathcal]{euscript}
\usepackage{mathrsfs}
\usepackage{dsfont}
\usepackage{mathtools}

\DeclarePairedDelimiter\floor{\lfloor}{\rfloor}

\newcommand{\ent}{\ensuremath{{\tt{h}}}}

\newcommand{\Kchannel}{{{\rm K}^{\Ldens{c}}}}
\newcommand{\Kcheck}{{{\rm K}^p}}
\newcommand{\Kvar}{{{\rm K}^v}}
\newcommand{\Kalt}{\lambda\Kvar}
\newcommand{\Kgen}{{\rm K}}
\newcommand{\Dd}{{\Ldens{D}}}
\RequirePackage{bbm}
\definecolor{darkgreen}{rgb}{0.14,0.5,0.14}

\newtheorem{theorem}{Theorem}
\newcommand{\btheo}{\begin{theorem}}
\newcommand{\etheo}{\end{theorem}}
\newcommand{\bproof}{\begin{proof}}
\newcommand{\eproof}{\end{proof}}
\newtheorem{definition}[theorem]{Definition}
\newcommand{\bdefi}{\begin{definition}}
\newcommand{\edefi}{\end{definition}}
\newtheorem{fact}[theorem]{Fact}
\newcommand{\bprop}{\begin{fact}}
\newcommand{\eprop}{\end{fact}}
\newtheorem{corollary}[theorem]{Corollary}
\newcommand{\bcor}{\begin{corollary}}
\newcommand{\ecor}{\end{corollary}}
\newtheorem{example}[theorem]{Example}
\newcommand{\bex}{\begin{example}}
\newcommand{\eex}{\end{example}}
\newtheorem{lemma}[theorem]{Lemma}
\newcommand{\blemma}{\begin{lemma}}
\newcommand{\elemma}{\end{lemma}}
\newtheorem{remark}[theorem]{Remark}
\newcommand{\bremark}{\begin{remark}}
\newcommand{\eremark}{\end{remark}}
\newtheorem{conj}[theorem]{Conjecture}
\newcommand{\bconj}{\begin{conj}}
\newcommand{\econj}{\end{conj}}

\newcommand{\naturals}{\ensuremath{\mathbb{N}}}

\newcommand{\expectation}{\ensuremath{\mathbb{E}}}

\newcommand{\prob}{\ensuremath{\mathbb{P}}}

\def\0{{\tt 0}} 
\def\1{{\tt 1}} 
\def\?{{\tt *}} 
\def\g{{\tt g}} 
\newcommand{\dee}{{\text d}}
\newcommand{\BPsmall}{\ensuremath{\text{\tiny BP}}} 
\newcommand{\qed}{{\hfill \footnotesize $\blacksquare$}}
\renewcommand{\mid}{\,|\,}

\newcommand{\dens}[1]{\mathsf{#1}}

\newcommand{\Ldens}[1]{\dens{#1}}

\newcommand{\Ddens}[1]{\mathfrak{{#1}}}

\newcommand{\BMS}{\ensuremath{\text{BMS}}}

\newcommand{\absDdist}[1]{\absd{\mathfrak{\MakeUppercase{#1}}}}
\newcommand{\absDdens}[1]{\absd{\Ddens{#1}}}

\newcommand{\absd}[1]{|#1|}

\newcommand{\dr}{d_r}
\newcommand{\dl}{d_l}

\DeclareMathOperator{\perr}{\mathfrak{E}}

\newcommand{\ind}{\mathbbm{1}}
\newcommand{\indicator}[1]{\ind_{\{ #1 \}}}

\DeclareMathOperator{\batta}{\mathfrak{B}}      
\newcommand{\entropy}{\text{H}}

\newcommand{\Lip}{\text{Lip}}

\newcommand{\vconv}{\circledast}
\newcommand{\cconv}{\boxast}

\newcommand{\SatLdens}[1]{\floor{\Ldens{#1}}}
\newcommand{\altSatLdens}[1]{\big\lfloor\!\!\lfloor \Ldens{#1}\rfloor\!\!\big\rfloor}
\newcommand{\sym}[1]{{#1_{\text{\tiny sym}}}}
\newcommand{\SatabsDdist}[2]{\absd{\lfloor\mathfrak{\MakeUppercase{#1}}\rfloor_{#2}}}
\newdimen\arrayruleHwidth
\makeatletter
\def\Hline{\noalign{\ifnum0=`}\fi\hrule \@height \arrayruleHwidth
   \futurelet \@tempa\@xhline}
	\makeatother
\allowdisplaybreaks
\begin{document}
\title{Analysis of Saturated Belief Propagation Decoding of Low-Density Parity-Check Codes}
\author{\IEEEauthorblockN{Shrinivas Kudekar, Tom Richardson and Aravind Iyengar \\ }
\IEEEauthorblockA{Qualcomm, New Jersey, USA\\ Email: \{skudekar,tomr,ariyengar\}@qti.qualcomm.com} \\
}

\maketitle
\thispagestyle{plain} 
\pagestyle{plain}
\begin{abstract}
We consider the effect of log-likelihood ratio saturation on belief propagation decoder low-density parity-check codes. 
Saturation is commonly done in practice and is known to have a significant effect on error floor performance.
Our focus is on threshold analysis and stability of density evolution.

We analyze the decoder for standard low-density parity-check code ensembles and show that 
belief propagation decoding generally degrades gracefully with saturation.
Stability of density evolution is, on the other hand, rather strongly effected by saturation and
the asymptotic qualitative effect of saturation is similar to reduction by one of variable node degree.

We also show under what conditions the block threshold
for the saturated belief propagation corresponds with the bit threshold.
\end{abstract}

\section{Introduction}  Standard belief propagation (BP) decoding of binary low-density parity-check (LDPC) codes involves passing messages typically representing log-likelihood ratios (LLRs) 
which can take any value in $\overline{\mathbb{R}} \triangleq
\mathbb{R}\cup \{\pm\infty\}$ \cite{RiU08}. The asymptotic analysis developed for BP decoding of LDPC codes inherently assumes that the messages have unbounded magnitude. In practice, however,
decoders typically use uniformly quantized and bound LLRs. 
Density evolution can be applied directly to such decoders but analysis
is often difficult and there are few general results.
Hence, it is of interest to understand the effect of saturation of LLR magnitudes
as a perturbation of full belief propagation. We call such a saturated decoder as a {\em saturating} belief propagation decoder (SatBP). Note that the decoder is strictly speaking not a BP decoder, but we adhere to the BP
nomenclature as we view SatBP as a perturbation of BP. 

In the design of capacity-achieving codes it is helpful to understand
how practical decoder concessions, like saturation, affect performance. For this purpose, we will analyze the SatBP decoder in the asymptotic limit of the blocklength going to infinity. 
In particular, if LLRs are saturated at magnitude $\Kgen$ then how much degradation
from the BP threshold should be expected.  
Naturally, one expects that as $\Kgen\to +\infty$, that one can reliably transmit
arbitrarily close to the BP threshold \cite{RiU08}.
We will see that this is not entirely correct and that, in particular, saturation can
undermine the stability of the perfect decoding fixed point if, for example, the fraction of degree two variable nodes in an irregular ensemble is non-zero. 
Our analysis shows that when the minimum variable node degree is at least three then there exists a large but finite saturation value $\Kgen$ such that the SatBP decoder can achieve arbitrarily small bit error rate whenever the full BP decoder can achieve arbitrarily small bit error rate.
Furthermore, a more careful stability analysis shows that in fact one can achieve reliability in terms of the block error rate.

\subsection{Related Work}
The papers \cite{6290251,6120169,5485006, 6567866} consider the effect of saturation on error floor performance.  It is observed in these works that saturation can limit the ability of decoding to escape trapping set behavior, thereby worsening error floor performance.
In \cite{6284049,6685976} some decoder variations are given that help reduce error floors.
Here we see an explicit effort to ameliorate the effect of saturation.
A related but distinct direction was taken in \cite{VNC14}.  There the authors made modifications to discrete node update rules so as to reduce error floor failure events. 
They fine tune finite state message update rules to optimize performance on a particular graph structure.
There have been other works that examine the effects of practical concessions.
In \cite{6134417} the authors consider the effect of quantization in LDPC coded flash memories.
In \cite{5624882} and\cite{6363268}  the effects of saturation and quantization are
modeled as noise terms.
Finally, in \cite{5205826} an analysis is done to evaluate the effect on capacity on quantization of
channel outputs. Although we take a different approach in this paper by focusing on asymptotic behavior, the fundamental conclusion is similar to the error floor results in \cite{6290251,6120169,5485006, 6567866}: saturation can dramatically effect the stability of the decoder. 

The paper is organized as follows. In the next section we will briefly review the standard asymptotic analysis of the BP decoder using density evolution (DE). 
Then in sections \ref{sec:SatBP} and \ref{sec:nonsymSatBP} we will introduce the SatBP decoder and perform perturbation analysis on the SatBP decoder using the Wasserstein metric \cite{KRU12}. 
In section \ref{sec:stabilityanalysis} we will use stability analysis to examine block thresholds for
SatBP.  We will see that in many cases the block threshold will correspond with the bit threshold,
but the conditions required are more stringent than in the non-saturated decoder case.

\section{BP decoding,  Density Evolution and the Wasserstein Distance}\label{sec:BPDEWass} In this
section we briefly review the BP decoder and the DE analysis \cite{RiU01} in the case of transmission over a
general BMS channel using standard LDPC code ensemble. Most of the material presented here can be
found in \cite{RiU08}.

We assume transmission over a BMS channel. Let $X (=\pm 1)$ denote the input and let $Y$
be the output. Further, let $p(Y=y \mid X=x)$ denote the {\em transition
probability} describing the channel. We generally characterize
a BMS channel
by its so-called $L$-distribution, $\Ldens{c}$. 
More precisely, $\Ldens{c}$ is the distribution of
\begin{align*} \ln \frac{p(Y \mid X= +1)}{p(Y \mid X=-1)} \end{align*}
conditioned that $X=+1.$  Generally, we may assume that 
\[
Y=\ln \frac{p(Y \mid X= +1)}{p(Y \mid X=-1)}\,.
\]
The symmetry of the channel is $p(Y=y \mid X=x) = p(Y=-y \mid X=-x)$
and the resulting densities $\Ldens{c}$ are symmetric, \cite{RiU08},
which means $e^{-\frac{1}{2}x}\Ldens{c}(x)$ is an even function of $x.$

Given $Z$ distributed according to $\Ldens{c}$, we write $\Ddens{c}$
to denote the distribution of $\tanh(Z/2),$ and
$\absDdens{c}$
to denote the distribution of $|\tanh(Z/2)|.$
We refer to these as the $\Dd$ and $|D|$ distributions respectively.
We use $\absDdist{c}$ to denote the corresponding cumulative $|D|$ distribution,
see \cite[Section~4.1.4]{RiU08}. 
Under symmetry, the distribution of $|Z|$ determines the distribution of $Z.$

For threshold analysis of LDPC ensembles we typically consider a
parameterized {\em family} of channels. We write $\{ \BMS(\sigma)\}$ to
denote the family parameterized by the scalar $\sigma$. Often it
will be more convenient to denote this family by $\{\Ldens{c}_\sigma\}$,
i.e., to use the family of $L$-densities which characterize the
channel family. 
One natural candidate for the parameter $\sigma$ is the 
entropy of the channel denoted by $\ent$. Thus, we also consider the characterization
of the family given by BMS($\ent$). 

\subsection{Degradation, Symmetric Densities and Functionals of Densities}
Let $p_{Z \mid X}(z\mid x)$ denote the transition probability
associated to a BMS channel $\Ldens{c}'$ and let $p_{Y \mid X}(y
\mid x)$ denote the transition probability of another BMS channel
$\Ldens{c}$.  We then say that $\Ldens{c}'$ is {\em degraded} with
respect to $\Ldens{c}$ if there exists a channel $p_{Z \mid Y}(z\mid
y)$ so that \begin{align*} p_{Z \mid X}(z \mid x) = \sum_{y} p_{Y \mid X}(y\mid
x) p_{Z \mid Y}(z\mid y).  \end{align*} We will use the notation
$\Ldens{c} \prec \Ldens{c}'$ to denote that $\Ldens{c}'$ is degraded
with respect to $\Ldens{c}$ (as a mnemonic think of $\Ldens{c}$ as the erasure
probability of a BEC and replace $\prec$ with $\leq$).

A useful characterization of degradation, see \cite{HaL69}, \cite[Theorem
4.74]{RiU08}, is that $\Ldens{c} \prec \Ldens{c}'$ is equivalent
to 
\begin{align}\label{equ:degradation} \int_0^1 f(x) \absDdens{c}(x)
\,\dee x \leq  \int_0^1 f(x) \absDdens{c'}(x) \,\dee x 
\end{align}
for all $f(x)$ that are non-increasing and concave on $[0,1]$.
In
particular, this characterization implies that $F(\Ldens{a})\leq
F(\Ldens{b})$ for $\Ldens{a} \prec \Ldens{b}$ if $F(\cdot)$ is
either the Battacharyya or the entropy functional.  This is true
since both are linear functionals of the distributions and their
respective kernels in the $|D|$-domain are decreasing and concave, see \cite{RiU08}.
An alternative characterization \cite{RiU08} of degradation in terms of the cumulative distribution
functions $\absDdist{c}(x)$ and $\absDdist{c'}(x)$ is that for all
$z \in [0, 1]$, 
\begin{align}\label{equ:degradationcdfs} \int_z^1
\absDdist{c}(x) \dee x \leq  \int_z^1 \absDdist{c'}(x) \,\dee x.
\end{align}

A BMS channel family $\{\BMS(\ent)\}_{\underline{\ent}}^{\overline{\ent}}$ is said to be {\em ordered} (by
degradation) if $\ent_1 \leq \ent_2$ implies $\Ldens{c}_{\ent_1}
\prec \Ldens{c}_{\ent_2}$.
(The reverse order, $\ent_1 \geq \ent_2,$ is also allowed but we
generally stick to the stated convention.)

\begin{definition}[Symmetric Densities]
Let $A$ denote an $L$-distribution in $\overline{\mathbb{R}} = \mathbb{R} \cup \{\pm\infty\}$.
Then $A$ is {\em symmetric} if it satisfies the following condition for every bounded, continuous function $f : \mathbb{R} \mapsto \mathbb{R}$,
\begin{equation} \label{eqn:lsym}
\int f(x)\mathrm{d}A(x) = \int e^{-x}f(-x)\mathrm{d}A(x).
\end{equation}

We say that an $L$-density $\Ldens{a}$ is {\em symmetric} if
$\Ldens{a}(-y) = \Ldens{a}(y) e^{-y}$. We recall that all densities
which stem from BMS channels are symmetric, see \cite[Sections
4.1.4, 4.1.8 and 4.1.9]{RiU08}. \qed 
\end{definition}

Functionals of densities often used in analysis are the Battacharyya, the entropy, and
the error probability functional.
For a density $\Ldens{a}$, these are denoted by $\batta(\Ldens{a})$,
$\entropy(\Ldens{a})$, and $\perr(\Ldens{a})$, respectively and are defined by
\begin{align*}
\batta(\Ldens{a}) & =\expectation( e^{-y/2}), \;\;
\entropy(\Ldens{a})  = \expectation( \log_2(1\!+\!e^{-y}))  \\
\perr(\Ldens{a}) & = \prob{\{y<0\}} + \frac{1}{2} \prob \{y=0\}.
\end{align*}
where $y$ is distributed according to $\Ldens{a}.$
Note that these definitions are valid even if $\Ldens{a}$ is not symmetric, although they
lose some of their original meaning.  We will apply these definitions to saturated densities
that are not necessarily symmetric. It is not hard to see that $\perr(\Ldens{a}) \leq \batta(\Ldens{a})$ for any density $\Ldens{a}$ not necessarily symmetric. Hence in the paper the main functional of interest is the Battarcharyya parameter.

%

\subsection{BP Decoder, DE analysis and the Wasserstein metric}

The definition of the standard BP decoder can be found in \cite{RiU08}. The
asymptotic performance of the BP decoder is given by the DE
technique \cite{RiU01, RiU08}.  Throughout the paper we will consider
standard  LDPC code ensembles as specified by their degree distributions \cite{RiU08}.
The analysis can be applied to more sophisticated structures, but we restrict to this case
for simplicity of presentation.
Thus we let $\lambda(\cdot)$ and $\rho(\cdot)$ represent the variable node and check node degree profile respectively. The ensemble is then denoted by $(\lambda, \rho)$. 

\begin{definition}[DE for BP Decoder cf. \cite{RiU08}]
For $\ell \geq 1$, the DE equation for a
$(\lambda, \rho)$ ensemble is given by
$$
\Ldens{x}_{\ell} = \Ldens{c} \vconv \lambda(\rho(\Ldens{x}_{\ell-1})).
$$
Here, $\Ldens{c}$ is the $L$-density of the BMS channel over which
transmission takes place and $\Ldens{x}_{\ell}$ is the density emitted by
variable nodes in the $\ell$-th round of density evolution. 
Initially we have $\Ldens{x}_{0}=\Delta_0$, the delta function at $0$. The operators
$\vconv$ and $\cconv$ correspond to the convolution of densities at
variable and check nodes, respectively, see \cite[Section 4.1.4]{RiU08}. The notation $\rho(\Ldens{x}_{\ell-1})$ represents the weighted check node convolution of the density $\Ldens{x}_{\ell-1}$. E.g., if $\rho(x) = x^{\dr-1}$, then $\rho(\Ldens{x}_{\ell-1}) = \Ldens{x}^{\cconv \dr-1}_{\ell-1}$.
\qed
\end{definition}
{\em Discussion}: For $(\dl, \dr)-$regular codes, the DE equation is given by $
\Ldens{x}_{\ell} = \Ldens{c} \vconv (\Ldens{x}^{\cconv \dr-1}_{\ell-1})^{\vconv \dl-1}.
$ The DE analysis is simplified when we consider the class of
symmetric message-passing decoders. The definition of symmetric message-passing
decoders can be found in \cite{RiU08}. Note that this definition of symmetry
pertains to the actual messages in the decoder and not to the densities which
appear in the DE analysis. We will see later that the saturated
decoder is a symmetric message-passing decoder and hence its DE analysis is
simplified by restricting to the use of the all zero (actually we use $+1$ for 'zero') codeword.

\begin{definition}[BP Threshold]
Consider an ordered and complete channel family $\{\Ldens{c}_\ent\}$.
Let $\Ldens{x}_{\ell}(\ent)$ denote the distribution in the $\ell$-th round
of DE when the channel is $\Ldens{c}_\ent$.
Then the {\em BP threshold} of the $(\lambda, \rho)$ ensemble 
is typically defined as
\begin{align*}
\ent^{\BPsmall}(\lambda, \rho, \{\Ldens{c}_\ent\}) & = \sup\{\ent: \Ldens{x}_{\ell}(\ent) \stackrel{\ell \to \infty}{\rightarrow} \Delta_{+\infty}\}.
\end{align*}
Here $\Delta_{+\infty}$ is the delta function at infinity representing the perfect decoding density.
An equivalent definition is
\begin{align*}
\ent^{\BPsmall}(\lambda, \rho, \{\Ldens{c}_\ent\}) & = \sup\{\ent: \perr(\Ldens{x}_{\ell}(\ent)) \stackrel{\ell \to \infty}{\rightarrow} 0\}.
\end{align*}
The later form is more convenient for our purposes and it is the one we shall adopt.
\qed
\end{definition}

We will also say that for a given channel $\Ldens{c}$, the BP decoder is {\em successful} if and only if $\perr(\Ldens{x}_{\ell}(\ent)) \stackrel{\ell \to \infty}{\rightarrow} 0$ or $\batta(\Ldens{x}_{\ell}(\ent)) \stackrel{\ell \to \infty}{\rightarrow} 0$. In other words, for any given $\epsilon > 0$, there exists $\ell$ such that $\batta(\Ldens{x}_{\ell}(\ent)) < \epsilon$.

In the sequel we will use the Wasserstein metric to measure distance between distributions.
We recall the definition of the Wasserstein metric below. For more properties of the Wasserstein metric see \cite{KRU11Wasserstein}.
\begin{definition}[Wasserstein Metric -- \protect{\cite[Chapter 6]{Villani09}}]\label{def:wasserstein}
Let $\absDdens{a}$ and $\absDdens{b}$ denote two $|D|$-distributions.
The Wasserstein  metric,
denoted by $d(\absDdens{a}, \absDdens{b})$, is defined as
\begin{align}\label{eq:blmetric}
d(\absDdens{a}, \absDdens{b})=\!\!\!\!\!\!\sup_{f(x) \in \Lip(1)[0, 1]} \!\Big\vert \int_{0}^{1} \!\!f(x)(\absDdens{a}(x)\!-\! \absDdens{b}(x)) \,\dee x \Big\vert,
\end{align}
where $\Lip(1)[0, 1]$ denotes the class of Lipschitz continuous functions on $[0, 1]$
with Lipschitz constant $1$.

In \cite{KRU12b} it is shown that the Wasserstein distance is equivalent to the
$L_1$ norm of the difference between the $|D|$-distributions.
\qed
\end{definition}

\section{Saturated Belief Propagation Decoding}\label{sec:SatBP}

In this section we introduce the saturated BP decoder. More precisely, 
we consider decoding with BP update rules at the nodes but the outgoing messages are restricted to the domain
$[-\Kgen ,\Kgen]$ for some $\Kgen > 0$ by saturation. 

\subsection{Saturated Decoder}
\begin{definition}[Saturation]
We define the \emph{saturation} operation at $\pm \Kgen$ for some $\Kgen \in \mathbb{R}^+$, denoted $\lfloor\cdot\rfloor_\Kgen$, by
\begin{equation} \label{eqn:clip}
\lfloor x \rfloor_\Kgen = \min(\Kgen, |x|)\cdot\mathrm{sgn}(x),
\end{equation}
where
\begin{align*}
\mathrm{sgn}(x) = \begin{cases}
-1, &x < 0\\
1, &x \geq 0
\end{cases}.
\end{align*}
\end{definition}

\begin{definition}[Saturated BP Decoder]\label{def:clippedBPdecoder}
Consider the standard $(\dl, \dr)$-regular ensemble. 
The saturated BP decoder is defined by the following rules. 
Let $\phi^{(\ell)}(\mu_1,\dots,\mu_{\dr-1})$ and $\psi^{(\ell)}(\mu_1, \dots,
\mu_{\dl-1})$ denote the outgoing message from the
check node and the variable node side respectively. Abusing the notation above, 
$\mu_1,\dots, \mu_{.}$ denotes the incoming messages on both the check node and
the variable node side. Then,
\begin{align*}
\phi^{(\ell)}(\mu_1,\dots,\mu_{\dr-1}) = \left\lfloor
2\tanh^{-1}\left(\prod_{i=1}^{\dr-1}\tanh(\mu_i/2)\right)\right\rfloor_\Kgen, \\
\psi^{(\ell)}(\mu_1,\dots,\mu_{\dl-1}) = \left\lfloor \mu_0 + \sum_{i=1}^{\dl-1}
\mu_i \right\rfloor_\Kgen,
\end{align*}
where $\mu_0$ is the message coming from the channel. Also, we set
$\phi^{(0)}(\mu_1,\dots,\mu_{\dr-1}) = 0$.
\end{definition}

\begin{lemma}[SatBP Decoder is symmetric]\label{lem:symmClipped}
The SatBP decoder given in Definition~\ref{def:clippedBPdecoder} is a symmetric message-passing decoder.
\end{lemma}
\begin{proof}
From Definition 4.83 in \cite{RiU08}  it is not hard to see that variable-node symmetry is satisfied for $\ell=0$. In general, variable node
symmetry is the following condition (for $\ell\geq 1$) on the message update function
\begin{align*}
 \psi^{(\ell)}(-\mu_0, -\mu_1, \dots, -\mu_{\dl-1}) 
 = -\psi^{(\ell)}(\mu_0, \mu_1, \dots, \mu_{\dl-1}). 
\end{align*}
Since $\lfloor x \rfloor_\Kgen = -\lfloor -x \rfloor_\Kgen$ we see that variable node
symmetry is preserved by saturation.
Let $b_1\in \{\pm 1\},\dots,b_{\dr-1} \in \{\pm 1\}$, then by Definition 4.83 in \cite{RiU08}, for the check node symmetry we have
\begin{align*}
& \phi^{(\ell)}(b_1\mu_1, \dots, b_{\dr-1}\mu_{\dr-1}) \\
& =  \min\Big(2\tanh^{-1}\!\!\Big(\prod_{i=1}^{\dr-1} \tanh(\vert\mu_i\vert/2)\Big), K\Big)\mathrm{sgn}\Big(\!\prod_{i=1}^{\dr-1} b_i \mu_i\Big) \\  
& = \min\Big(2\tanh^{-1}\!\!\Big(\prod_{i=1}^{\dr-1} \tanh(\vert\mu_i\vert/2)\Big), K\Big)\mathrm{sgn}\Big(\!\prod_{i=1}^{\dr-1} \mu_i\!\Big)\!\! \prod_{i=1}^{\dr-1}\!\! b_i \\
& = \phi^{(\ell)}(\mu_1, \dots, \mu_{\dr-1})  \Big(\prod_{i=1}^{\dr-1} b_i\Big).
\end{align*}
and we see again that symmetry is preserved by saturation.
\end{proof}

{\em Discussion:} The symmetry of the message-passing decoder together with symmetry of the channel allows us to use the all-zero codeword assumption. This along with the concentration results (see
Theorem 4.94 in \cite{RiU08}) allows to write down the density evolution of the
SatBP decoder in the usual way.
Note that if messages entering a check node are saturated in magnitude at $\Kgen$ then outgoing messages are automatically saturated at $\Kgen.$ This holds not just for BP but for many message passing algorithms such as the min-sum algorithm. Our analysis has two parts: bounding the effect of saturation over finitely many iterations and stability analysis.  For the bounding analysis we focus on BP although the technique can be easily extended to other decoders.  In the stability analysis we explicitly relax the assumptions to cover a variety of check node updates.

Given $X \sim \Ldens{a}$, let $\SatLdens{a}_\Kgen$ denote the distribution of 
$\lfloor X \rfloor_\Kgen.$  Note that the saturation operation can be viewed as a channel taking
$X$ to $\lfloor X \rfloor_\Kgen.$  We have immediately
\[
\Ldens{a} \prec \SatLdens{a}_\Kgen\,.
\]
In general $\SatLdens{a}_\Kgen$ will not be symmetric even if $\Ldens{a}$ is symmetric
since we will not typically have $\SatLdens{a}_\Kgen(-\Kgen)= e^{-\Kgen} \SatLdens{a}_\Kgen(\Kgen).$
If $\Ldens{a}$ is symmetric then we will have
\begin{equation}
\SatLdens{a}_\Kgen(-\Kgen)  \le  e^{-\Kgen} \SatLdens{a}_\Kgen(\Kgen).
\end{equation}

Although using lemma~\ref{lem:symmClipped} one can write down the DE recursion
for the SatBP decoder, we know that in general the densities will not be
symmetric. Two of the most useful properties of DE for BP are that it preserves both symmetry of 
densities and ordering by degradation.  These properties are sacrificed by saturation, but can be recovered
with a slight variation.  There are two alternatives for this.  One is to place the saturated probability mass at
$\pm z$ instead at $\pm \Kgen$ where $z$ is chosen according to the actual LLR conditioned on magnitude $\Kgen.$
The second alternative is to slightly degrade the density by moving some probability mass from
$\Kgen$ to $-\Kgen.$  This can be interpreted operationally as flipping the sign of a message with magnitude
$\Kgen$ with some probability $\gamma.$  The flipping rate $\gamma$ is chosen so that the resulting 
probability that the sign of the message is incorrect is $e^{-\Kgen}/(1+e^{-\Kgen}).$  In general $\gamma$
is upper bounded by this value and for large $\Kgen$ this is a small perturbation.  Of the two approaches
the second is inferior in that it degrades the channel more than the first. On the other hand,
the second approach preserves ordering by degradation while the first does not. We shall adopt the 
second approach.

Let us introduce the notation $D(p,z)$ to denote the density
\[
D(p,z) = p \Delta_{-z} + (1-p) \Delta_{z}\,.
\]
Here $\Delta_z$ ($\Delta_{-z}$ ) is the delta function at $z$ ($-z$).
We  will sometimes denote $p \Delta_{-z}$ as $D_-(p,z)$ and
$(1-p) \Delta_{z}$ as $D_+(p,z).$
When $(p,z)$ is clear from context we may drop it from the notation.
Using this notation we have for symmetric $\Ldens{a},$
\begin{align}\label{eq:satDensrepresentation}
\SatLdens{a}_\Kgen = \gamma D(q,z)(x) + \Ldens{a}(x)\mathds{1}_{\{|x| < \Kgen\}}
\end{align}
where $\gamma = \prob_{\Ldens{a}} \{ |x| \ge \Kgen \}$ and $\gamma q = \prob_{\Ldens{a}} \{ x \le -\Kgen \}.$
\begin{lemma}[Symmetric SatBP] \label{lem:symclipprop}
Given a symmetric density $\Ldens{a}$ we define
\[
\SatLdens{a}_\sym{\Kgen} = \gamma D(p,z)(x) + \Ldens{a}(x)\mathds{1}_{\{|x| < \Kgen\}}
\]
where $p = e^{-\Kgen}/(1+e^{-\Kgen})$ and $\gamma =  \prob_{\Ldens{a}} \{ |x| \ge \Kgen \}.$
Then,
\begin{enumerate}[(i)]
\item $\SatLdens{a}_\sym{\Kgen}$ is a symmetric $L$-density.
\item $\SatLdens{a}_\Kgen \prec \SatLdens{a}_\sym{\Kgen}$.
\end{enumerate}
\end{lemma}
\begin{proof}
Part (i) is immediate. To prove part (ii) we note that comparing with the non-symmetrized case we see that 
\( p \ge q \,.\)  Thus, $\SatLdens{a}_\sym{\Kgen}$ can be realized by taking messages with distribution
$\SatLdens{a}_\Kgen$ and flipping the sign of a message with magnitude $\Kgen$ by a quantity $\lambda$
with $\lambda$ determined by
\[
p = \frac{e^{-\Kgen}}{1+e^{-\Kgen}} = \lambda (1-q) + (1-\lambda) q\,.
\]
\end{proof}
As a consequence of Lemma \ref{lem:symclipprop}, we will term the operation used
to obtain $\SatLdens{a}_\sym{\Kgen}$ from $\Ldens{a}$ as \emph{symmetric-saturation}. 

We summarize all the claims above in the following.
\begin{corollary}[Degradation Order]\label{lem:degradationOrder}
For symmetric $\Ldens{a}$ we have  
\[
\Ldens{a} \prec \SatLdens{a}_\Kgen \prec \SatLdens{a}_\sym{\Kgen}.
\]
\end{corollary}
It is fairly intuitive that as $\Kgen$ becomes larger, the density $\SatLdens{a}_\sym{\Kgen}$ should become close to the density $\Ldens{a}$. This is the content of the next lemma which uses the Wasserstein distance between distributions.
\begin{lemma}\label{lem:DistSymClip}
Let $\Ldens{a}$ be a symmetric $L$-density. 
Then,
\begin{align*}
d(\Ldens{a}, \SatLdens{a}_\sym{\Kgen}) \leq 1 - \tanh({\Kgen}/2),
\end{align*}
where $d(\cdot, \cdot)$ is the Wasserstein distance defined previously. 
\end{lemma}
\begin{proof}
For any $0 \le z < \Kgen$ we have 
\(
\prob_{\Ldens{a}} \{ x \le z \}
=
\prob_{\SatLdens{a}_\Kgen} \{ x \le z \}
=
\prob_{\SatLdens{a}_\sym{\Kgen}} \{ x \le z \}
\)
and for any $z \ge \Kgen$ we have 
\(
1
=
\prob_{\SatLdens{a}_\Kgen} \{ x \le z \}
=
\prob_{\SatLdens{a}_\sym{\Kgen}} \{ x \le z \}\,.
\)
Since $\tanh(x/2)$ is increasing and $\tanh(-x/2)=-\tanh(x/2)$ we have
\[
\SatabsDdist{a}{\sym{\Kgen}} (z) =
\indicator{z<\tanh(\Kgen/2)} \absDdist{a} (z)
+
\indicator{z\ge\tanh(\Kgen/2)}\,.
\]
By \cite{KRU12b}, we have that the Wasserstein distance is equivalent to the
$L_1$ norm of the difference between the $|D|$-distributions.
Clearly, the distance is
bounded by $1 - \tanh (\Kgen/2)$.
%
\end{proof}

Let $T(\cdot)$ denote a DE iteration for the full BP decoder, i.e.,
\[
T(\Ldens{c}, \Ldens{x}) = \Ldens{c} \vconv \lambda(\rho(\Ldens{x})).
\]

\begin{definition}[DE for Sym. and Non-Sym. Saturation]
Consider a BMS channel with $L$-density $\Ldens{c}$. Let $\Delta_0$ denote the perfectly noisy channel.
Let $\Ldens{x}^{(0)} = \Delta_0$. Then the DE for symmetric SatBP decoder is defined as, 
\begin{align*}
\Ldens{x}^{(\ell)} = \lfloor\Ldens{c}\vconv \lambda(\rho(\Ldens{x}^{(\ell-1)}))\rfloor_\sym{\Kgen}.
\end{align*}
The DE for non-symmetric SatBP decoder is defined as,
\begin{align*}
\Ldens{x}^{(\ell)} = \lfloor\Ldens{c}\vconv \lambda(\rho(\Ldens{x}^{(\ell-1)}))\rfloor_{\Kgen}.
\end{align*}
Finally, we use the notation
$S_\sym{\Kgen}(\Ldens{c}, \Ldens{x}) = \left\lfloor T(\Ldens{c}, \Ldens{x})\right\rfloor_\sym{\Kgen}$ and 
$S_\Kgen(\Ldens{c}, \Ldens{x}) = \left\lfloor T(\Ldens{c}, \Ldens{x})\right\rfloor_{\Kgen}$. \qed
\end{definition}
Now imagine that we run both the full DE and symmetric saturated DE starting with the density $\Delta_0$. In the next lemma we show that at every iteration the order of degradation between the full DE and symmetric saturated DE is preserved. We will use the notation $T^{(\ell)}(\Ldens{c}, \Delta_0)$ to denote the $\ell$ iteration of the full DE. More precisely, $T^{(\ell)}(\Ldens{c}, \Delta_0) = T(\Ldens{c},T^{(\ell-1)}(\Ldens{c}, \Delta_0))$. As a shorthand, we will use $T^{(\ell)}(\Ldens{c}, \Delta_0) = T(T^{(\ell-1)}(\Ldens{c}, \Delta_0))$. We similarly define $S_\sym{\Kgen}^{(\ell)}(\Ldens{c}, \Delta_0)$ and $S_{\Kgen}^{(\ell)}(\Ldens{c}, \Delta_0)$.
\begin{lemma}[Degradation Order under DE]\label{lem:DEvsSymClipping}
For any $\ell \ge 0$ we have 
\[
T^{(\ell)}(\Ldens{c}, \Delta_0)  \prec S_\sym{\Kgen}^{(\ell)}(\Ldens{c}, \Delta_0).
\]
\end{lemma}
\begin{proof}
Let $\Ldens{x}^{(\ell)}$ denote the DE for usual BP decoder and
 $\Ldens{z}^{(\ell)}$ denote the DE for the symmetric saturation operation. Since
$\Ldens{x}^{(0)} = \Ldens{z}^{(0)} = \Delta_0$, we have that $\Ldens{x}^{(1)} =
\Ldens{c}$ and $\Ldens{z}^{(1)} = \SatLdens{c}_\sym{\Kgen}$. From
corollary~\ref{lem:degradationOrder} we get that $\Ldens{x}^{(1)} \prec \Ldens{z}^{(1)}$. 
Now, since DE preserves the order of degradation, we get
$$
\Ldens{x}^{(2)} = T(\Ldens{c}, \Ldens{x}^{(1)})   \prec T(\Ldens{c}, \Ldens{z}^{(1)}) \stackrel{\text{Lem.}\,\ref{lem:degradationOrder}}{\prec} \Big\lfloor T(\Ldens{c}, \Ldens{z}^{(1)})\Big\rfloor_\sym{\Kgen} = \Ldens{z}^{(2)}.$$
Continuing, for all $\ell$ we get $T^{(\ell)}(\Ldens{c}, \Delta_0) \prec S_\sym{\Kgen}^{(\ell)}(\Ldens{c}, \Delta_0)$.
\end{proof}
We now estimate the distance between the densities appearing in the DE of standard
BP and the DE of the symmetric saturation operation. For this we again use the
Wasserstein distance defined previously (for symmetric densities). 
\begin{lemma}[Distance Between Symmetric SatBP and BP]\label{lem:distsymclippedandBP}
Consider $\ell$ iterations of the DE for the standard BP and the symmetric saturation operation.
Then
\begin{align*}
d(T^{(\ell)}&(\Ldens{c}, \Delta_0), S_\sym{\Kgen}^{(\ell)}(\Ldens{c}, \Delta_0)) \leq 2e^{-\Kgen + \ell\cdot\ln (2(\dl-1)(\dr-1))}.
\end{align*}
\end{lemma}

\begin{proof}
Let $T(\cdot)$ and $S_\sym{\Kgen}(\cdot)$ be defined as in lemma~\ref{lem:DEvsSymClipping}
and consider the Wasserstein distance between them. We have,
\begin{align}
d&(T^{(\ell)}(\Ldens{c}, \Delta_0), S_\sym{\Kgen}^{(\ell)}(\Ldens{c}, \Delta_0)) \notag \\
&= d(T(T^{(\ell - 1)}(\Ldens{c}, \Delta_0)), S_\sym{\Kgen}(S_\sym{\Kgen}^{(\ell - 1)}(\Ldens{c}, \Delta_0))) \notag \\
&\stackrel{\text{Trian. ineq.}}{\leq} d(T(T^{(\ell - 1)}(\Ldens{c}, \Delta_0)), T(S_\sym{\Kgen}^{(\ell - 1)}(\Ldens{c}, \Delta_0))) \notag \\
&\qquad+ d(T(S_\sym{\Kgen}^{(\ell - 1)}(\Ldens{c}, \Delta_0)), S_\sym{\Kgen}(S_\sym{\Kgen}^{(\ell - 1)}(\Ldens{c}, \Delta_0))) \notag \\
&\stackrel{\text{(viii), Lem. 13 in \cite{KRU12b}}}{\leq} \alpha_{\ell} d(T^{(\ell - 1)}(\Ldens{c}, \Delta_0), S_\sym{\Kgen}^{(\ell - 1)}(\Ldens{c}, \Delta_0)) \notag \\
&\qquad+ d(T(S_\sym{\Kgen}^{(\ell - 1)}(\Ldens{c}, \Delta_0)), S_\sym{\Kgen}(S_\sym{\Kgen}^{(\ell - 1)}(\Ldens{c}, \Delta_0))) \notag \\
&\stackrel{(a)}{\leq} \alpha_{\ell} d(T^{(\ell - 1)}(\Ldens{c}, \Delta_0), S_\sym{\Kgen}^{(\ell - 1)}(\Ldens{c}, \Delta_0)) + \left(1 - \tanh\Big(\frac{\Kgen}{2}\Big)\right), \notag
\end{align}
where
\begin{align*}
\alpha_{\ell} & = 2 (\dl-1)  
 \sum_{j\!=\!1}^{\dr\!-\!1} (1\!-\!\batta^2(\Ldens{a}))^{\frac{\dr\!-\!1\!-\!j}2}(1\!-\!\batta^2(\Ldens{b}))^{\frac{j\!-\!1}2},
\\&
\le 2 (\dl-1) (\dr-1)
\end{align*}
where $\Ldens{a} = T^{(\ell - 1)}(\Ldens{c}, \Delta_0)$ and $\Ldens{b} = S_\sym{\Kgen}^{(\ell - 1)}(\Ldens{c}, \Delta_0)$ and $\dl$ and $\dr$ correspond to the average variable node and check node degrees. 
Also, the inequality $(a)$ is obtained by using lemma~\ref{lem:DistSymClip}. 

Continuing with the above inequality we get,
\begin{align}
d(&T^{(\ell)}(\Ldens{c}, \Delta_0), S_\sym{\Kgen}^{(\ell)}(\Ldens{c}, \Delta_0)) \notag \\
&\leq  (1\! - \!\tanh\Big(\frac{\Kgen}{2}\Big))(1 \!+\! \alpha_{\ell} \!+\! \alpha_{\ell}\alpha_{\ell-1} \!+\! \dots \!+\! \alpha_{\ell}\alpha_{\ell-1}\cdots \alpha_2), \notag
\end{align}
From the bound on $\alpha_{\ell}$ we obtain $(1 \!+\! \alpha_{\ell} \!+\! \alpha_{\ell}\alpha_{\ell-1} \!+\! \dots \!+\! \alpha_{\ell}\alpha_{\ell-1}\cdots \alpha_2) \leq  (2(\dl-1)(\dr-1))^{\ell}$.

Combining with $1 - \tanh(\Kgen/2) \leq 2e^{-\Kgen}$  we get,
\begin{align*}
d(&T^{(\ell)}(\Ldens{c}, \Delta_0), S_\sym{\Kgen}^{(\ell)}(\Ldens{c}, \Delta_0)) \leq  2e^{-\Kgen + \ell\cdot\ln (2(\dl-1)(\dr-1))}.
\end{align*}
\end{proof}

The above gives us a bound on the $\batta(S_\sym{\Kgen}^{(\ell)}(\Ldens{c}, \Delta_0))$. Using (ix) Lemma~13 
in \cite{KRU12b} we get,
\begin{align}
\batta(S_\sym{\Kgen}^{(\ell)}&(\Ldens{c}, \Delta_0)) \leq  \batta(T^{(\ell)}(\Ldens{c}, \Delta_0)) \nonumber \\ 
& + 2\sqrt{d(T^{(\ell)}(\Ldens{c}, \Delta_0), S_\sym{\Kgen}^{(\ell)}(\Ldens{c}, \Delta_0))} \nonumber \\
& \leq \batta(T^{(\ell)}(\Ldens{c}, \Delta_0)) + 2\sqrt{2}e^{\frac{-\Kgen + \ell\cdot\ln (2(\dl-1)(\dr-1))}2}.\label{eqn:SKsymTdiff}
\end{align}

{\em Discussion}: 
In the sequel, we will denote $\Kgen$ by $\Kvar$ to distinguish between the saturation levels appearing at variable nodes, check nodes and the channel. To summarize, we show that for $\Kvar$ large enough, for every iteration the Battacharyya parameter of the symmetric saturated DE remains close to the Battacharyya of the full DE. In the next section we will relate the symmetric saturated DE to the non-symmetric saturated DE to show that the Battacharyya parameter for the SatBP decoder can also be made small by choosing $\Kvar$ large enough.

\section{Convergence of Nonsymmetrized Saturated DE}\label{sec:nonsymSatBP}
The results of the previous section show that, when transmitting below the threshold of the full BP decoder and using sufficiently many iterations, the Battacharrya parameter of the densities in the symmetric SatBP decoder can be small by choosing $\Kvar$ large enough. More precisely, consider transmission over a general BMS channel $\Ldens{c}$ such that we are transmitting below the BP threshold of the channel family. Let us assume transmission using $(\lambda,\rho)$ ensemble with average variable node and check node degree given by $\dl$ and $\dr$ respectively. Then, given an $\epsilon >0$, there exists $\ell_0(\Ldens{c}, \epsilon)\in \naturals$ such that for all $\ell \geq \ell_0$, $\batta(T^{(\ell)}(\Ldens{c}, \Delta_0)) \leq \epsilon/2$. 
Then, by choosing $\Kvar$ large enough, specifically $\Kvar > \Kgen_0 \triangleq l_0(\Ldens{c},\epsilon) \ln (2(\dl-1)(\dr-1)) + 2\ln \frac{4\sqrt{2}}{\epsilon}$, we have that $\batta(S_\sym{\Kgen}^{(\ell)}(\Ldens{c}, \Delta_0))  \leq \epsilon$.

\subsection{Non-symmetrized SatBP Decoder}
We now show that the Battacharrya parameter for the non-symmetric SatBP decoder can also be made small by choosing $\Kvar$ large enough. We first consider a fixed computation tree and then average over the tree ensemble. 

We begin with an operational description of symmetrization.
Consider a fixed tree $\sf{T}$ of depth $\ell.$
Let $Y$ denote the vector of received LLR values associated to the variable nodes
under the all-zero codeword assumption.
In addition, for each variable node we assume an independent random variable uniformly
distributed on $[0,1].$  We denote the vector of these variables by $Z = \{Z_v\},$ where $v$ is the index for the variable nodes.
Now, the node operations correspond to BP except that outgoing messages from 
the variable nodes are magnitude saturated at $\Kvar.$
The independent random variables are used for the flipping operation.
The flipping probability for each node is determined by density evolution.
If the outgoing message has magnitude $\Kvar$ then its sign is flipped if $Z_v < \lambda_v$
where $\lambda_v$ is the appropriate flipping probability.

Let the received LLR magnitude of a variable node $v$ be $x$. 
The probability with which we flip the bit is such that the final error probability is equal to $\frac{e^{-\Kvar}}{1+e^{-\Kvar}}$. For received LLR magnitude of $x$, the probability that it is received correctly is $\frac1{1+e^{-x}}$. As a consequence we get,
$$
\frac{e^{-\Kvar}}{1+e^{-\Kvar}} = \lambda_v \frac{1}{1 + e^{-x}} + (1 - \lambda_v) \frac{e^{-x}}{1 + e^{-x}},
$$
where $\lambda_v$ is the flipping probability of variable node $v$ and $x \geq \Kvar$.
Solving we get $\lambda_v = \frac{e^{-\Kvar}}{1 + e^{-\Kvar}} \frac{1 - e^{-x + \Kvar}}{1 - e^{-x}} \leq \frac{e^{-\Kvar}}{1 + e^{-\Kvar}}.$ Thus the probability that a variable node, with a received LLR magnitude greater than $\Kvar$, is not flipped is at least $\frac1{1 + e^{-\Kvar}} \geq 1 - e^{-\Kvar}$.

Let us denote the outgoing message at the variable node by $x$. From the above we see that 
the distribution of the outgoing message $x$ is $S_\sym{\Kgen}^{(\ell)}(\Ldens{c}, \Delta_0)).$
Let us consider the conditional distribution $p(x\mid Y,Z).$
We obtain $S_\sym{\Kgen}^{(\ell)}(\Ldens{c}, \Delta_0))$ by averaging over $Y$, $Z$ and the code ensemble.
Let $A_\Kvar$ denote the event that $Z_v \ge 1 - e^{-\Kvar}$ for each $v.$  This is clearly independent of the received values.
Assuming a fixed computation tree $T$ (i.e., we suppress dependence on $T$ in the notation) we have
\[
p(x\mid Y ) = p(x\mid Y,A_\Kvar) p(A_\Kvar) + p(x\mid Y,\bar{A}_\Kvar)(1-p(A_\Kvar)),
\]
where $\bar{A}_\Kvar$ denotes the complement event and, by independence, we can 
averaging over $Y$ to obtain
\[
p(x) = p(x\mid A_\Kvar) p(A_\Kvar) + p(x\mid \bar{A}_\Kvar)(1-p(A_\Kvar))
\]
hence
\[
p(x\mid A_\Kvar) = \frac{p(x) -  p(x\mid \bar{A}_\Kvar)(1-p(A_\Kvar))}{p(A_\Kvar)}
\]
Now $p(x\mid A_\Kvar)$ is the distribution of the non-symmetric SatBP decoder.
Intuitively one expects $p(x\mid \bar{A}_\Kvar)$ to be inferior (higher probability of error,
larger Battacharyya parameter) to $p(z\mid {A}_\Kvar),$ but this appears difficult to prove.
We have, however, $p(A_\Kvar) \ge (1-e^{-\Kvar})^{|V(\sf{T})|} \ge  1-e^{-\Kvar}{|V(\sf{T})|}$ where ${|V(\sf{T})|}$ is the number of variable nodes in the tree.

The above analysis is summarized in the following lemma.
\begin{lemma}[SatBP Decoder versus Symmetrized SatBP]\label{lem:clippeddecodervssymclipping}
For any $0 < \epsilon < 1$ and $\ell \in \naturals$, there exists a $\Kvar$ large enough
such that
$$
\batta(S^{(\ell)}_\Kgen(\Ldens{c}, \Delta_0)) \leq \frac1{1-\epsilon}\batta(S^{(\ell)}_\sym{\Kgen}(\Ldens{c}, \Delta_0)).
$$
\end{lemma}
\begin{proof}
From the above analysis we have that for a fixed tree $\sf{T}$ of depth $\ell$,
\begin{align*}
p(x\mid A_\Kvar) & = \frac{p(x) -  p(x\mid \bar{A}_\Kvar)(1-p(A_\Kvar))}{p(A_\Kvar)} \\
& \leq \frac{p(x)}{p(A_\Kvar)} \leq \frac{p(x)}{1 - e^{-K}{|V(\sf{T})|}}. 
\end{align*}
where $p(x\mid A_\Kvar)$ is the distribution of the non-symmetric SatBP decoder.  
For any fixed number of iterations, the total maximum number of variable nodes in a computation tree is fixed. Hence we can take $\Kvar$ large enough so that
$e^{-\Kvar}{|V(\sf{T})|} < \epsilon$ for all $T.$  Note that the required $\Kvar$ grows linearly in the number of iterations.
Averaging over the tree ensemble and multiplying by the kernel $e^{-x/2}$, we get the desired result.
\end{proof}
{\em Discussion:} Let us summarize. From the above analysis we have that for any $0 < \epsilon < 1/2$, there exists $\Kvar > 0$, large enough such that the Battacharyya parameter of the SatBP decoder is upper bounded by $\epsilon$. Note that the value of $\Kvar$ depends on the number of iterations of the full BP required to get its Battacharyya parameter to be at the most $\epsilon/2$. So given a channel $\Ldens{c}$ such that the BP decoder is successful when transmitting over $\Ldens{c}$, the number of such iterations required is fixed. Call it $\ell_0(\Ldens{c}, \epsilon)$. Then, from the above analysis we have that  for $\Kvar \geq \Kgen_0  \triangleq l_0(\Ldens{c},\epsilon) \ln (2(\dl-1)(\dr-1)) + 2\ln \frac{8\sqrt{2}}{\epsilon}$,  $\batta(S_{\Kgen}^{(\ell)}(\Ldens{c}, \Delta_0))  \leq \epsilon$. 
Note that we can make the Battacharyya as small as desired by increasing the number of iterations and consequently increasing $\Kvar$. But then the saturation value  $\Kvar$ becomes infinite. Hence to make the Battacharyya arbitrarily small we now need to show that once the Battacharyya parameter is made small enough, by choosing $\Kvar$ large but fixed, then the subsequent iterations of the SatBP decoder will drive the Battacharyya parameter down to zero. This is the content of the stability analysis done in the next section. We will see that in order to make the Battacharyya parameter arbitrarily small, it is sufficient to bring it close to the stability region. By choosing $\epsilon$ according to equation \eqref{eq:nearstab3} and arguments following it, we can choose $\Kvar$ large enough so that we are  guaranteed to be in the stability region. Furthermore, we have that $\Kgen_0$, defined above, now depends only on the channel $\Ldens{c}$ and the degree distribution. 

\section{Stability Analysis of the SatBP Decoder}\label{sec:stabilityanalysis}
An important part of the asymptotic analysis of LDPC codes involves the analysis of the convergence of DE to a zero error state.  In this section we analyze the stability of the SatBP.
We begin with some necessary conditions.

For stability of the zero error condition there must exist a positive invariant set
of zero error distributions, i.e.,  a subset $\mathcal{S}$ of distributions
so that $\perr{(\Ldens{s})} =0$ for all $\Ldens{s} \in \mathcal{S}$ and
$S_{\Kgen}(\Ldens{c}, \Ldens{s}) \in \mathcal{S}.$
Existence of $\mathcal{S}$ follows easily from the compactness of the 
space of densities and continuity of DE.

\begin{lemma}\label{lem:necessarysupport}
Assume the channel $\Ldens{c}$ has support at $-L,$ $L>0.$
In an irregular ensemble with minimum variable degree $d_l$
the support of all densities in $\mathcal{S}$ must lie in 
$[L/(\dl-2),\infty].$
\end{lemma}
\begin{IEEEproof}
It is obvious that $\mathcal{S} = \emptyset$ in an irregular ensemble with $d_l=1,$
so we assume $d_l \ge 2.$
We use $\Ldens{a}^{(\ell)}$ and $\Ldens{b}^{(\ell)}$ to denote the density of the message coming out of the variable nodes and check nodes respectively in the density evolution process. 
We claim that if $\Ldens{a}^{(\ell)}$ has support on $(-\infty,z_\ell]$ with $z_\ell > 0$ then
$\Ldens{a}^{(\ell+1)}$ has support on $(-\infty,z_{\ell+1}]$ with $z_{\ell+1} = z_\ell - (L-(\dl-2) z_\ell).$
To see the claim note that $\Ldens{b}^{(\ell)}$ also has support on $(-\infty,z_\ell]$
and it follows that $\Ldens{a}^{(\ell+1)}$ has support on $(-\infty,z_{\ell+1}]$
where  $z_{\ell+1} = (\dl-1) z_\ell - L = z_\ell - (L-z_\ell (\dl-2)).$

Assume $\Ldens{a}^{(0)}\in \mathcal{S}$ has support on $(-\infty,z_0]$
where $z_0 < L/(\dl-2)$ and define $\delta := L  - (\dl-2) z_0 > 0.$
By the above claim it follows from an inductive argument
that $\Ldens{a}^{(\ell)}\in \mathcal{S}$ has support on $(-\infty,z_\ell]$
where $z_\ell$ is a decreasing sequence satisfying
 $z_{\ell} \le z_{0} -  \ell \delta.$
For $\ell$ large enough the right hand side is negative, implying a non-zero error probability, and we obtain a contradiction with the definition of $\mathcal{S}$.
\end{IEEEproof}

\subsection{Failure of Stability with Degree Two}

From Lemma \ref{lem:necessarysupport} we immediately have 
\begin{lemma}
In an irregular ensemble with $\lambda_2>0$ no invariant set $\mathcal{S}$
exists for any value of $\Kvar<\infty$ unless the channel is the BEC.
\end{lemma}
\begin{IEEEproof}
If $d_l = 2$ and the channel is not the BEC and hence has support on $(-\infty,0),$
then Lemma \ref{lem:necessarysupport}
shows that there can be no positive invariant zero-error set of distributions with support
on $[-\Kvar,\Kvar]$ for $\Kvar<\infty.$
\end{IEEEproof}

In the case of the BEC it can be seen that saturated DE matches unsaturated DE except that the mass
at $+\infty$ in unsaturated DE is not placed at $+\Kvar.$  Hence, stability is unaffected by saturation.
If the channel has unbounded support on $(-\infty,0]$, then there is no possibility of stability under saturation no matter
what the degree.
A condition on the  finite channel support is given in the section on stability with degree at least three.

\subsection{Near Stability}\label{sec:nearstab}

Even though stability with saturation cannot be achieved in irregular ensembles with
degree two variable nodes, it is not surprising that for large $\Kvar$ the residual error rate
can be made very small.  For sufficiently large $\Kvar$ the residual error rate will have
no practical consequence.  In this section we quantify the residual error rate.

The stability analysis of standard irregular ensembles under BP decoding
rests on the relations
\begin{equation}\label{eqn:DEvnodebatta}
\batta(\Ldens{c}\vconv \lambda{(\Ldens{a})}) = \batta{(\Ldens{a})} \lambda{(\Ldens{a})}
\end{equation}
and
\begin{equation}\label{eqn:DEcnodebatta}
\batta{(\rho{(\Ldens{a})})} \le 1 - \rho(1-\batta{(\Ldens{a})})\,.
\end{equation}
Equality \eqref{eqn:DEvnodebatta} continues to hold without symmetry of 
$\Ldens{a}$ or $\Ldens{c}.$
The inequality \eqref{eqn:DEcnodebatta}, however, does not hold without symmetry.
In Appendix~\ref{app:checknodebatta} we prove a more general form of the following.
\begin{lemma} \label{lem:additivecheckbound}
Let the incoming L-densities at a degree $d+1$ check node be
$\Ldens{a}_1,...,\Ldens{a}_d$ and let $\Ldens{b}$ be the outgoing density.
Then
\[
\batta{(\Ldens{b} )} \le
\sum_{i=1}^d \batta{(\Ldens{a}_i)} \,.
\]
\end{lemma}
{\em Discussion:} The above result holds for a wide range of check node update operations including
BP and the min-sum decoder.  

Throughout this section we will use $\Ldens{a}$ ($\Ldens{b}$) to denote the density coming out of a variable node (check node).  We also use $\Ldens{a}^{(n)}$  and $\Ldens{b}^{(n)}$ to denote the densities coming out of the variable nodes and check nodes at the $n$th iteration of the saturated DE recursion. We prove the following result,
\begin{lemma}\label{lem:nearstability}
Consider an irregular ensemble with minimum variable node degree $d_{\text{min}}\ge 2.$
Assume $\lambda_2 \rho'(1) \batta{(\Ldens{c})} < 1.$
Then, there exists a constant $x^*,$  a constant $N,$
and a constant $C(d_{\text{min}})$
such that, for all $\Kvar$ large enough, if for some $n_0$ we have
$\batta{(\Ldens{a}^{(n_0)})} \le x*$ then
$\batta{(\Ldens{a}^{(n)})} \le C(d_{\text{min}}) e^{-{\Kvar}/2}$
for all $n \ge n_0+N.$
Moreover, if $d_{\text{min}} > 2$ we can have $C(d_{\text{min}})=3.$
\end{lemma}
\begin{proof}
To incorporate saturation into the analysis based on
the Battacharyya parameter we have the inequality for any $\Kgen>0$,
\[
\batta{(\SatLdens{a}_\Kgen)}\le \batta{(\Ldens{a})} + e^{-{\Kgen}/2}.
\]
Indeed, we have
\begin{align}\label{eq:satbattainq}
\batta{(\SatLdens{a}_\Kgen)} & = e^{{\Kgen}/2}\int_{-\infty}^{-\Kgen}\!\!\! \Ldens{a}(x) d x + \int_{-\infty}^{+\infty} \!\!\!\!\indicator{\vert x \vert < \Kgen}\Ldens{a}(x) e^{-x/2} d x \nonumber \\
& \quad \quad \quad \quad+  e^{-{\Kgen}/2}\int_{\Kgen}^{\infty}\!\!\! \Ldens{a}(x) d x \nonumber \\
& \leq \int_{-\infty}^{-\Kgen}\!\!\!\! \Ldens{a}(x) e^{-x/2} d x + \int_{-\infty}^{+\infty} \!\!\!\!\indicator{\vert x \vert < \Kgen}\Ldens{a}(x) e^{-x/2} d x \nonumber \\
& \quad \quad\quad + \int_{\Kgen}^{\infty} \Ldens{a}(x) e^{-x/2} d x +  e^{-{\Kgen}/2} 
\nonumber \\
& = \batta(\Ldens{a}) + e^{-{\Kgen}/2},
\end{align}
where the last inequality follows since $e^{-{\Kgen}/2}\int_{\Kgen}^{\infty} \Ldens{a}(x) d x \leq e^{-{\Kgen}/2}\int_{-\infty}^{\infty} \Ldens{a}(x) d x  = e^{-{\Kgen}/2}$.
As a result of the saturation of messages, we see that the minimum value of the Battacharyya parameter is equal to $e^{-{\Kgen}/2}$ and we can therefore not hope
to reach a smaller value.

{\em Minimum variable node degree equal to 2:}
Let us assume $d_{\text{min}} = 2,$ i.e., $\lambda_2 > 0.$
Let $\Ldens{a}^{(n_0)}$ be any $L$-density which need not be symmetric.
Consider 
\[
g(x):= \lambda_2  \batta(\Ldens{c}) \rho'(1) + (1-\lambda_2)\batta(\Ldens{c}) (\rho'(1))^2 x  \,.
\]
Since $\lambda_2  \batta(\Ldens{c}) \rho'(1) < 1$, there exists an $x^*>0$ such that $g(x^*) < 1$. Choose $x^*$ such that $g(x^*) < 1$ and $\rho'(1) x^*<1$.  Now assume $\batta(\Ldens{a}^{(n_0)}) \leq x^*$. 
Choose $\Kvar$ large enough such that $ \frac1{1-g(x^*)} e^{-{\Kvar}/2} < x^*$.

Let us perform the saturated DE recursion once. We have,
\begin{align}\label{eq:battarecursiondeg2}
\batta(&\Ldens{a}^{(n_0+1)})  =  \batta( \lfloor \Ldens{c} \vconv \lambda(\rho(\Ldens{a}^{(n_0)}))\rfloor_\Kvar) \nonumber \\
& \stackrel{\eqref{eq:satbattainq}}{\leq}  \batta(\Ldens{c} \vconv \lambda(\rho(\Ldens{a}^{(n_0)}))) + e^{-{\Kvar}/2} \nonumber \\ 
& = \batta(\Ldens{c})\lambda\Big(\sum_i \rho_i \batta((\Ldens{a}^{(n_0)})^{\cconv (i-1)})\Big) + e^{-{\Kvar}/2} \nonumber \\
& \stackrel{\text{Lemma}~\ref{lem:additivecheckbound}}{\leq}
 \batta(\Ldens{c})\lambda\Big(\batta(\Ldens{a}^{(n_0)})\sum_i (i-1)\rho_i \Big) + e^{-{\Kvar}/2} \nonumber \\
& \stackrel{\text{since}\, \rho'(1) \batta{(\Ldens{a}^{(n_0)})} < 1}{\leq} \lambda_2  \batta(\Ldens{c}) \rho'(1) \batta{(\Ldens{a}^{(n_0)})} \nonumber \\
&+ (1-\lambda_2)\batta(\Ldens{c}) (\rho'(1) \batta{(\Ldens{a}^{(n_0)}}))^2
+ e^{-{\Kvar}/2} \nonumber \\
&=g\big(\batta{(\Ldens{a}^{(n_0)}})\big) \batta{(\Ldens{a}^{(n_0)}})+ e^{-{\Kvar}/2}  \nonumber \\
&\le g(x^*) \batta{(\Ldens{a}^{(n_0)}})+ e^{-{\Kvar}/2} \\
&\le g(x^*) x^*+ e^{-{\Kvar}/2}  \nonumber \\
&\le  x^* \nonumber,
\end{align}
where the last inequality follows from the choice of $\Kvar$.

By induction, the above inequality gives $\batta{(\Ldens{a}^{(n)}}) \leq x^*$ for all $n \geq n_0$.
Consider any $n = n_0+k$. Also by induction on
 \eqref{eq:battarecursiondeg2}, we get
\begin{align*}
\batta(\Ldens{a}^{(n_0+k)}) & \leq x^*  (g(x^*))^k
+ e^{-{\Kvar}/2} \sum_{j=0}^{k-1}  (g(x^*))^j\\
& = x^* (g(x^*))^k+ e^{-{\Kvar}/2}\frac{1- (g(x^*))^k}{1-g(x^*)}\,.
\end{align*}
It follows that any $\epsilon>0$ and all $k$ large enough we have
\begin{align*}
\batta(\Ldens{a}^{(n_0+k)}) & \leq   e^{-{\Kvar}/2}\frac{1- \epsilon}{1-g(x^*)}\,.
\end{align*}


{\em Minimum variable node degree equal to 3:}
Let us now assume that the minimum variable node degree is 3. 
Let us denote,
\begin{align}\label{eq:nearstab3}
f(x) = \lambda_3  \batta(\Ldens{c}) \rho'(1)^2 x
+ (1-\lambda_3)\batta(\Ldens{c}) \rho'(1)^3 x^2.
\end{align}
Choose $x^*>0$ such that $f(x^*) \leq 1/2$ and $\rho'(1) x^* < 1$. Let $n_0$ be such that $\batta(\Ldens{a}^{(n_0)}) \leq x^*$. Choose $\Kvar$ large enough so that $2e^{-{\Kvar}/2} < x^*.$ 
Following the previous analysis, we have for all $n \geq n_0$
\begin{align*}
\batta{(\Ldens{a}^{(n+1)}}) \le &
\lambda_3  \batta(\Ldens{c}) (\rho'(1) \batta{(\Ldens{a}^{(n)})})^2 \\
&+ (1-\lambda_3)\batta(\Ldens{c}) (\rho'(1) \batta{(\Ldens{a}^{(n)}}))^3
+ e^{-{\Kvar}/2}
\end{align*}

A little algebra then shows that there exists $N > n_0$ so that for all $n\ge N$  we have
\begin{align}
\batta{(\Ldens{a}^{(n)})} & \le 3e^{-{\Kvar}/2} \label{vnodeallbnd}\\
\batta{(\Ldens{b}^{(n)})} & \le  3\rho'(1) e^{-{{\Kvar}/2}}\label{checkallbnd}
\end{align}
where $\Ldens{b}^{(n)}$ denotes the density coming out of the check nodes. Also, \eqref{checkallbnd} follows from \eqref{vnodeallbnd} and Lemma \ref{lem:additivecheckbound}.
\end{proof}
The ``near stability'' analysis done above can clearly not show convergence to zero error although it can be used
to show convergence to relatively small error rate. 
As we showed above, unlike the unsaturated case, zero error rate convergence cannot be achieved
with the saturated decoder when degree two variable nodes are included.
For degree three and higher, stability can be shown but a refined analysis is needed.


\subsection{Stability Analysis with Minimum Variable Node Degree Equal to Three}\label{sec:stabmindegree3}

In this section we consider irregular ensembles where the minimum variable node degree
is at least three.  We generalize the standard stability analysis by separating out the
saturated probability mass and tracking it through the variable node and check node updates.
For simplicity we shall restrict to right regular ensembles.
We show that convergence to zero error rate occurs and that convergence is
exponential in iteration.  In the unsaturated case this can be achieved with
degree two variable nodes and with degree three and above doubly exponential
convergence occurs.  In subsequent sections we show that double exponential convergence
can be attained in the saturated case for degree four and above although a modification is needed for degree four.
For degree three doubly exponential convergence can be recovered but only with the dramatic
and likely impractical step of erasing all received values near the end of the decoding.

We assume regular check nodes with degree $d_r$ and we let
$\Kcheck$ denote the magnitude of an outgoing message
when all incoming messages have magnitude $\Kvar.$  Although we focus on BP-like decoding
our analysis applies to other algorithms such as min-sum, in which case we have $\Kcheck=\Kvar.$
In general, if $\Kgen_1,...,\Kgen_{d_r-1}$ are incoming message magnitudes at a check node then we assume that
the corresponding outgoing magnitude $\Kgen_{\text{out}}$ satisfies
\begin{align}\label{eqn:magbound}
- \ln \sum_{i=1}^{d_r-1} e^{-\Kgen_i} \le \Kgen_{\text{out}} \le \min_i \{ \Kgen_i \}
\end{align}
Both conditions are satisfied by BP and min-sum. E.g., for BP we can write explicitly $\tanh (\Kgen_i/2) = (1 - e^{-\Kgen_i/2})/(1 + e^{-\Kgen_i/2})$ and then some algebra\footnote{Indeed, it is not hard to see that $\frac{1 - e^{-\Kgen_{\text{out}}}}{1 + e^{-\Kgen_{\text{out}}}} = \frac{1 - \sum_i e^{-\Kgen_i} + A}{1 + \sum_i e^{-\Kgen_i} + B}$, where $A, B \geq 0$. Furthermore, one can show that $A(1 + \sum_ie^{-\Kgen_i}) \geq B (1 - \sum_ie^{-\Kgen_i})$, which implies that $\frac{1 - e^{-\Kgen_{\text{out}}}}{1 + e^{-\Kgen_{\text{out}}}} \geq \frac{1 - \sum_i e^{-\Kgen_i}}{1 + \sum_i e^{-\Kgen_i}}$ giving us the inequality. } gives us \eqref{eqn:magbound}. 
We note in passing that the left inequality implies
\(
- \ln \sum_{i=1}^{d_r-1} e^{-\lambda \Kgen_i} \le \lambda \Kgen_{\text{out}} 
\)
for all $\lambda \in [0,1].$  We will make use of the case $\lambda = \frac{1}{2}.$

Messages entering a check node update $\Ldens{a}$ have the form 
\[
\Ldens{a} = \gamma D(p,\Kvar) +\bar{\gamma}\Ldens{m}
\]
where $\Ldens{m}$ is supported on $(-\Kvar,\Kvar)$ and has total mass $1$
(if it has zero probability we have $\bar{\gamma}=0.$)

Messages entering a variable node update $\Ldens{b}$ have the form 
\[
\Ldens{b} = \gamma D(p,\Kcheck) +\bar{\gamma}\Ldens{m}
\]
where $\Kcheck\le \Kvar$ is the outgoing magnitude at a check when all incoming magnitudes equal $\Kvar$
and $\Ldens{m}$ is supported on $(-\Kcheck,\Kcheck).$ 
From \eqref{eqn:magbound} we have $e^{-\Kcheck} \le (d_r-1) e^{-\Kvar}.$ 
We assume $\Kvar >2 \ln(d_r-1)$ large enough so that $2\Kcheck > \Kvar.$
In the subsequent analysis we also assume that the support of the channel $\Ldens{c}$ is restricted to $(-\Kchannel, \Kchannel)$ where we assume that $\Kchannel \leq 2\Kcheck - \Kvar .$

The analysis tracks the  quantities $\gamma p$ and $\bar{\gamma}\batta{(\Ldens{m})}.$
For stability we aim to show that both quantities converge to $0.$
Note that this implies that $\gamma \rightarrow 1.$
In the standard stability analysis of irregular ensembles and full BP, one tracks the Battacharyya parameter of the
density through the DE iterations when the density is near $\Delta_\infty.$
At the check node the Battacharyya parameter undergoes a constant factor gain with
a factor of $\rho'(1).$  On the variable node side the parameter is raised to the power of
the minimum variable node degree less one, and scaled the channel Battacharyya.
Thus, one arrives at the stability condition $\lambda_2 \rho'(1) \batta{(\Ldens{c})} < 1.$
If the minimum variable node degree is three then the update bound takes the form
$\batta{(\Ldens{a}^{(\ell+1)})} \le C \batta{(\Ldens{a}^{(\ell)})}^2$, for some positive constant $C$, and one obtains
doubly exponential decay in $\batta{(\Ldens{a}^{(\ell)})}.$
For the saturated case we accomplish something similar, although the conditions are different.
As a first step we show that we still have constant factor gain at check nodes.

\subsubsection{Check Node Analysis}\label{sec:checknodestabilityanalysis}
We assume a right regular ensemble with check degree $d+1.$
Let us represent the density entering the check node as
$\gamma D(p,\Kvar) + \bar{\gamma}\Ldens{m}$
where $\Ldens{m}$ is a density supported on $(-\Kvar,\Kvar).$
Then the density emerging out of the check node is given by 
$ \gamma' D(p',\Kcheck) + \bar{\gamma'}\Ldens{m}' \triangleq (\gamma D(p,\Kvar) + \bar{\gamma}\Ldens{m})^{\cconv d}
$, where $\Kcheck$ is the magnitude of the check output when all inputs are $\Kvar,$
which satisfies $\Kvar-\ln d \le \Kcheck \le \Kvar,$ and support of $\Ldens{m}'$ is also $(-\Kcheck, \Kcheck)$. Let us now perform the computation explicitly. In this section we use $\Dd$ to denote $D(p,\Kvar)$.
We have,
\begin{align*}
&(\gamma D(p,\Kvar) + \bar{\gamma}\Ldens{m})^{\cconv d}
=
\sum_{k=0}^{d} \binom{d}{k} \gamma^k\bar{\gamma}^{d-k}\Dd^{\cconv k} \cconv \Ldens{m}^{\cconv d-k}
\\&=\bar{\gamma}^{d}\Ldens{m}^{\cconv d}+
\sum_{k=1}^{d-1} \binom{d}{k} \gamma^k\bar{\gamma}^{d-k}\Dd^{\cconv k} \cconv \Ldens{m}^{\cconv d-k}
+ \gamma^d \Dd^{\cconv d} 
\end{align*}
where we have separated out two of the terms from the sum.
Although we have indicated that density evolution for check node update
is associative, which it is for min-sum and sum-product algorithms,
we do not actually require the associative property and a density
$\Dd^{\cconv k} \cconv \Ldens{m}^{\cconv d-k}$ can simply be understood as the outgoing one corresponding to $k$ incoming messages from
density $\Dd$ and $d-k$ messages from density $\Ldens{m}.$

By Lemma \ref{lem:checkbattabound} we have for $1 \leq k \leq d-1,$
\begin{align*}
\batta{(\Dd^{\cconv k} \cconv \Ldens{m}^{\cconv d-k})} &\le
(1+k (e^{\frac{\Kvar}{2}}\batta{(\Dd)}-1)) (d-k) \batta{(\Ldens{m})}\,
\\&\le
k e^{\frac{\Kvar}{2}}\batta{(\Dd)} (d-k) \batta{(\Ldens{m})}\,.
\end{align*}
A little algebra shows that
\begin{align*}
\sum_{k=1}^{d-1} \binom{d}{k} \gamma^k\bar{\gamma}^{d-k}k(d-k)
=
\gamma\bar{\gamma}d(d-1)
\end{align*}
and we now obtain
\begin{align*}
&\batta{\Bigl(
\sum_{k=1}^{d-1} \binom{d}{k} \gamma^k\bar{\gamma}^{d-k}\Dd^{\cconv k} \cconv \Ldens{m}^{\cconv d-k}
\Bigr)}
\\ & \le
\gamma\bar{\gamma}d(d-1) e^{\frac{\Kvar}{2}}\batta{(\Dd)} \batta{(\Ldens{m})}\,.
\end{align*}
Lemma \ref{lem:checkbattabound} also gives
\[
\batta{(\Ldens{m}^{\cconv d})} \le
d \batta{(\Ldens{m})}\,,
\]
so we now have
\begin{align*}
&\batta{\Bigl(
\sum_{k=0}^{d-1} \binom{d}{k} \gamma^k\bar{\gamma}^{d-k}\Dd^{\cconv k} \cconv \Ldens{m}^{\cconv d-k}
\Bigr)}
\\ & \le
d\Bigl(
(d-1) \gamma e^{\frac{\Kvar}{2}}\batta{(\Dd)} +1
\Bigr)
\bar{\gamma}\batta{(\Ldens{m})}\,.
\end{align*}
We have $\gamma' D(p',\Kcheck) = \gamma^d \Dd^{\cconv d}$
so $p' = \frac{1-(1-2p)^d}{2} \le d p$ where we have used
Lemma \ref{lem:pprodineq} to  obtain the last inequality. 

We summarize the results as follows.
\begin{lemma}\label{lem:checknodeanalysis}
Let the incoming density to a degree $d+1$ check node be
$\gamma D(p,\Kvar) + \bar{\gamma}\Ldens{m}.$
Then the outgoing density
$\gamma' D(p',\Kcheck) + \bar{\gamma'}\Ldens{m'}$
satisfies the following
\begin{align*}
\begin{bmatrix} 
\bar{\gamma}' \batta(\Ldens{m'}) \\  \gamma' p'
\end{bmatrix}
 \le
d
\begin{bmatrix} 
\xi  & 0\\ 
0 & 1 
\end{bmatrix}
\begin{bmatrix} 
\bar{\gamma} \batta(\Ldens{m}) \\  \gamma p
\end{bmatrix}
\end{align*}
where
\( \xi =\Bigl(
(d-1) \gamma e^{\frac{\Kvar}{2}}\batta{(D(p,\Kvar))} +1
\Bigr)  \,.\) 
\end{lemma}
In the stability region we will have the bound $\xi \le 3$ so we see
that we have been able to obtain a linear growth bound for the
check node density evolution update.

\subsubsection{Variable Node Analysis}\label{sec:stabmindegree3varnodeanalysis}
\newcommand{\nti}{{n_{-}}}
\newcommand{\ntj}{{n_{+}}}
\newcommand{\ntm}{{n_{\Ldens{m}}}}

Consider a variable node of degree $d+1$ and incoming density
\[
\Ldens{b} = \gamma D(p,\Kcheck) +\bar{\gamma}\Ldens{m}.
\]
The outgoing density from the variable node has the form 
\[
\Ldens{a} =\gamma' D(p',\Kvar) +\bar{\gamma}'\Ldens{m'}.
\]
The density $\Ldens{a}$ is the saturation of
\begin{align}
\begin{split}
\sum_{k=0}^{d-2}& \binom{d}{k}\gamma^{k}\bar{\gamma}^{d-k}\Ldens{c}\vconv \Dd^{\vconv k}\vconv \Ldens{m}^{\vconv (d-k)}
\\&+d\bar{\gamma} \gamma^{d-1}\Ldens{c}\vconv \Dd^{\vconv {d-1}} \vconv \Ldens{m} + \gamma^{d}\Ldens{c}\vconv \Dd^{\vconv d}
\end{split}\label{eqn:threetermsa}
\end{align}
where in this section we use $\Dd$ to denote $D(p,\Kcheck).$
In particular $\gamma' p'$ is the total mass of this density on $(-\infty,-\Kvar]$
and $\gamma' \Ldens{m}'$ is the restriction of this density to $(-\Kvar,\Kvar).$

We see in the above decomposition that incoming messages either have magnitude $\Kcheck,$ i.e. are drawn from $\Dd,$
or they are drawn from $\Ldens{m}$
and therefore take values in $(-\Kcheck,\Kcheck).$  We can define a type for an outgoing message
consisting of a triple of non-negative integers $(\nti,\ntm,\ntj)$ where $\nti+\ntm+\ntj=d.$
Here
$\nti$ represents the number of $-\Kcheck$ incoming messages, $\ntj$ the number of $+\Kcheck$ incoming messages,
and $\ntm$ the number of incoming message drawn from $\Ldens{m}$ that comprise the outgoing message.
Our analysis will pay special attention to the terms with $\ntm=0$ and $\ntm=1$
which is why we distinguished these terms.

A handy elementary result is the following.
\begin{lemma}\label{lem:binomtrick}
If $a,b \ge 0$ and $k \le d$ then 
\[
\sum_{i=0}^{d-k} \binom{d}{i} a^{d-i} b^i \le \binom{d}{k} a^k (a+b)^{d-k}
\]
\end{lemma}
\begin{IEEEproof}
For $i \leq d - k$ we have,
\begin{align*}
\binom{d}{i} \le \binom{d}{i}\binom{d-i}{k} =  \binom{d}{k}\binom{d-k}{i}\,.
\end{align*}
and the lemma follows from the binomial theorem.  We remark that 
there is an alternate form since
\(
 \binom{d}{k} = \binom{d}{d-k}\,.
\)
\end{IEEEproof}

Let us consider the three parts of \eqref{eqn:threetermsa}.
The first part comprises messages types $(\nti,\ntm,\ntj)$ where
$\ntm \ge 2.$
The second part comprises messages types $(\nti,\ntm,\ntj)$ with
$\ntm =1$
and the third part comprises messages types $(\nti,\ntm,\ntj)$ with
$\ntm =0.$
We will consider the contribution of each part to $\gamma' p'$
and to $\bar{\gamma}' \Ldens{m}'.$

Let us first consider $\gamma' p'.$  We use the bound $\int_{-\infty}^{-\Kgen} \Ldens{a}(x) dx \le e^{-\frac{\Kgen}{2}}\batta(\Ldens{a})$, which is valid for any
density and any $\Kgen \ge 0$, Lemma \ref{lem:binomtrick} and the multiplicative property of Battacharyya parameter at the variable node side to obtain
\newcommand{\omp}{\bar{p}}
\begin{align}
\int_{-\infty}^{-\Kvar}& \sum_{k=0}^{d-2} \binom{d}{k}\gamma^{k}\bar{\gamma}^{d-k}\Ldens{c}\vconv \Dd^{\vconv k}\vconv \Ldens{m}^{\vconv (d-k)}  (x) dx 
\nonumber
\\&\le
e^{-\frac{\Kvar}{2}} \frac{d(d-1)}{2} (\bar{\gamma}\batta(\Ldens{m}))^2 \batta(\Ldens{c})\batta(\Ldens{b})^{d-2}\,.\label{eqn:gp2}
\end{align}
Now we consider contributions from $\ntm=1.$
A message of type $(\nti,1,\ntj)$ has value at most
$(\ntj-\nti)\Kcheck + (\Kcheck+\Kchannel)$
and at least
$(\ntj-\nti)\Kcheck - (\Kcheck+\Kchannel).$ Recall that $(-\Kchannel, \Kchannel)$ is the channel support. 
Hence if $\ntj-\nti > 0$ then the message has value greater than $-\Kvar$
and  if $\ntj-\nti < -1$ then the message has value less than $-\Kvar.$
If $\ntj-\nti = 0$ then the message has value less than $-\Kvar$ only if the
contribution from $\Ldens{c}\vconv\Ldens{m}$ is less than $-\Kvar.$
If $\ntj-\nti = -1$ then the message can have value less than $-\Kvar$
only if the
contribution from $\Ldens{c}\vconv\Ldens{m}$ is less than $0.$
Hence, we obtain
\begin{align}
\begin{split}
\int_{-\infty}^{-\Kvar}& \Ldens{c} \vconv \Ldens{m} \vconv \Dd^{d-1} (x) dx 
\le
\\&
\begin{cases}
\sum_{j=0}^{\frac{d-4}{2}} \binom{d-1}{j} p^{d-1-j} \omp^j
\\ \quad+  \binom{d-1}{\frac{d-2}{2}}p^{\frac{d}{2}} \omp^{\frac{d-2}{2}}\perr(\Ldens{c}\vconv\Ldens{m})
& d \text{ even}  \\
\sum_{j=0}^{\frac{d-3}{2}} \binom{d-1}{j} p^{d-1-j} \omp^j 
\\ \quad+  \binom{d-1}{\frac{d-1}{2}}p^{\frac{d-1}{2}} \omp^{\frac{d-1}{2}}
e^{-\frac{\Kvar}{2}}\batta(\Ldens{c}\vconv\Ldens{m})
& d \text{ odd}  
\end{cases}
\end{split}\label{eqn:gp1}
\end{align}
Note that for the case $d$ even, we use $\perr(\Ldens{c}\vconv\Ldens{m})$ to bound the contribution from $(\Ldens{c}\vconv\Ldens{m})(x)$ for $x\leq 0$.
Now we consider contributions from $\ntm=0.$
A message of type $(\nti,0,\ntj)$ has value at most
$(\ntj-\nti)\Kcheck + (\Kchannel)$
and at least
$(\ntj-\nti)\Kcheck - (\Kchannel).$
Hence if $\ntj-\nti \ge 0$ then the message has value greater than $-\Kvar$
and  if $\ntj-\nti < -1$ then the message has value less than $-\Kvar.$
If $\ntj-\nti = -1$ then the message can have value less than $-\Kvar$
only if the
contribution from $\Ldens{c}$ is less than $0.$
Hence, we obtain
\begin{align}
\begin{split}
\int_{-\infty}^{-\Kvar}&  \Ldens{c} \vconv \Dd^d (x) dx 
\le
\\&
\begin{cases}
\sum_{j=0}^{\frac{d-2}{2}} \binom{d}{j} p^{d-j} \omp^j
& d \text{ even}  \\
\sum_{j=0}^{\frac{d-3}{2}} \binom{d}{j} p^{d-j} \omp^j 
+  \binom{d}{\frac{d-1}{2}}p^{\frac{d+1}{2}} \omp^{\frac{d-1}{2}} \perr(\Ldens{c})
&  d \text{ odd}  
\end{cases}
\end{split}\label{eqn:gp0}
\end{align}

Using the bound $\perr(\Ldens{c}\vconv\Ldens{m}) \le \batta(\Ldens{c}\vconv\Ldens{m})$ and
Lemma \ref{lem:binomtrick} we obtain from \eqref{eqn:gp1}
\begin{align*}
\int_{-\infty}^{-\Kvar} &\Ldens{c} \vconv \Ldens{m} \vconv \Dd^{d-1} (x) dx 
\le \\  &
\begin{cases}
  \binom{d-1}{\frac{d-2}{2}}p^{\frac{d}{2}} (p+ \batta(\Ldens{c}\vconv\Ldens{m}))
& d \text{ even}  \\
  \binom{d-1}{\frac{d-1}{2}}p^{\frac{d-1}{2}} (
p+
e^{-\frac{\Kvar}{2}}\batta(\Ldens{c}\vconv\Ldens{m})
)
& d \text{ odd}  
\end{cases}
\end{align*}
and using the bound $\perr(\Ldens{c}) \le 1$ and
Lemma \ref{lem:binomtrick} we obtain from \eqref{eqn:gp0}
\begin{align*}
&\int_{-\infty}^{-\Kvar}  \Ldens{c} \vconv \Dd^d (x) dx 
\le
\binom{d}{\lfloor\frac{d-1}{2}\rfloor} p^{\lceil \frac{d+1}{2}\rceil}\,.
\end{align*}
Combining the above into \eqref{eqn:threetermsa} we have
\begin{align}
\begin{split}
\gamma' p' \le&e^{-\frac{\Kvar}{2}} \frac{d(d-1)}{2} (\bar{\gamma}\batta(\Ldens{m}))^2 \batta(\Ldens{c})\batta(\Ldens{b})^{d-2}
\\&+
 d  \binom{d-1}{\lfloor\frac{d-1}{2}\rfloor}(\gamma p)^{\lfloor\frac{d}{2}\rfloor} \bigl((\gamma p)+ \batta(\Ldens{c})(\bar{\gamma}\batta(\Ldens{m}))\bigr)
\\&+
\binom{d}{\lfloor\frac{d-1}{2}\rfloor} (\gamma p)^{\lceil \frac{d+1}{2}\rceil}\,
\end{split}\nonumber\\
\begin{split}
\\ \le
&e^{-\frac{\Kvar}{2}} \frac{d(d-1)}{2} (\bar{\gamma}\batta(\Ldens{m}))^2 \batta(\Ldens{c})\batta(\Ldens{b})^{d-2}
\\&+
(d+1) (4 \gamma p)^{\lfloor\frac{d}{2}\rfloor+1}\,+
 d(4 \gamma p)^{\lfloor\frac{d}{2}\rfloor}  \batta(\Ldens{c})(\bar{\gamma}\batta(\Ldens{m}))
\end{split}\label{eqn:gppbound}
\end{align}
where  we have used $\binom{d}{\lfloor\frac{d-1}{2}\rfloor} \le 2^{d-1}.$ 
We note that when $d$  is odd we can add another factor of $e^{-\frac{\Kvar}{2}}$ to the last term.

Now we consider the contribution to  $\bar{\gamma}' \Ldens{m}'.$ 
Let us introduce the notation $\SatLdens{a}_\Kgen^\circ(x) = \Ldens{a}(x)\indicator{|x|<\Kgen}.$
First we note that the contribution to $ \batta(\Ldens{m}')$ from types with $\ntm \geq 2$ is upper bounded by
\begin{align*}
&\batta\bigl( \sum_{k=0}^{d-2} \binom{d}{k}\gamma^{k}\bar{\gamma}^{d-k}\Ldens{c}\vconv \Dd^{\vconv k}\vconv \Ldens{m}^{\vconv (d-k)} \bigr)
\le
\\&
 \frac{d(d-1)}{2} (\bar{\gamma}\batta(\Ldens{m}))^2 \batta(\Ldens{c})\batta(\Ldens{b})^{d-2},
\end{align*}
where we applied Lemma \ref{lem:binomtrick}. 

Let us introduce the notation $q = e^{\frac{\Kcheck}{2}}p$ and $\tilde{q} =e^{-\frac{\Kcheck}{2}}\omp.$ Note that for any density $\Ldens{a}$ we have
$\batta (\Ldens{a}\vconv \Delta_\Kgen) = e^{-\Kgen}\batta(\Ldens{a}).$

Now we consider the contribution from types with $\ntm=1.$
A type $(\nti,1,\ntj)$ will have a non-zero contribution only if the interval centered on
$(\ntj-\nti)\Kcheck$ of width $2(\Kchannel+\Kcheck)$ intersects $(-\Kvar,\Kvar).$
Note that $\Ldens{m}' = \SatLdens{\Ldens{c} \vconv \Ldens{m} \vconv \Dd^{d-1}}_{\Kvar}^\circ$.
Since we assume $2\Kcheck  \ge \Kchannel+\Kvar$ and $\Kcheck \le \Kvar$ we obtain
\begin{align*}
&\batta(\SatLdens{\Ldens{c} \vconv \Ldens{m} \vconv \Dd^{d-1}}_{\Kvar}^\circ)
\le
\\&
\batta(\Ldens{c})\batta(\Ldens{m})
\begin{cases}
\sum_{j=\frac{d-2}{2}}^{\frac{d}{2}} \binom{d-1}{j} q^{d-1-j} \tilde{q}^j
& d \text{ even}  \\
\sum_{j=\frac{d-3}{2}}^{\frac{d+1}{2}} \binom{d-1}{j} q^{d-1-j} \tilde{q}^j 
& d \text{ odd}  
\end{cases}
\end{align*}
Using the inequality 
\(
2\binom{d-1}{\frac{d-3}{2}} \ge \binom{d-1}{\frac{d-1}{2}}
\)
for odd $d$ we can write this as

\begin{align*}
&\batta(\SatLdens{\Ldens{c} \vconv \Ldens{m} \vconv \Dd^{d-1}}_{\Kvar}^\circ)
\\ \le &
\batta(\Ldens{c})\batta(\Ldens{m})
\begin{cases}
\binom{d-1}{\frac{d}{2}} (q\tilde{q})^{\frac{d-2}{2}} (q+\tilde{q})
& d \text{ even}  \\
 \binom{d-1}{\frac{d-3}{2}} (q\tilde{q})^{\frac{d-3}{2}}  (q+\tilde{q})^2
& d \text{ odd}  
\end{cases}
\\ = &
\batta(\Ldens{c})\batta(\Ldens{m})
\begin{cases}
\binom{d-1}{\frac{d}{2}} (p\omp)^{\frac{d-2}{2}}\batta(D(p,\Kcheck))
& d \text{ even}  \\
 \binom{d-1}{\frac{d-3}{2}} (p\omp)^{\frac{d-3}{2}}  \batta(D(p,\Kcheck))^2
& d \text{ odd}  
\end{cases}
\\ \le &
\batta(\Ldens{c})\batta(\Ldens{m})
\begin{cases}
 (4p)^{\frac{d-2}{2}}\batta(D(p,\Kcheck))
& d \text{ even}  \\
2 (4p)^{\frac{d-3}{2}}  \batta(D(p,\Kcheck))^2
& d \text{ odd}  
\end{cases}
\end{align*}

Finally we consider the contribution from types with $\ntm=0.$
A type $(\nti,0,\ntj)$ will have a non-zero contribution only if the interval centered on
$(\ntj-\nti)\Kcheck$ of width $2\Kchannel$ intersects $(-\Kvar,\Kvar).$
Hence we obtain
\begin{align*}
& \batta(\SatLdens{\Ldens{c}  \vconv \Dd^{d}}_{\Kvar}^\circ)
\\ \le&
\batta(\Ldens{c})\begin{cases}
 \binom{d}{\frac{d}{2}} q^{\frac{d}{2}} \tilde{q}^{\frac{d}{2}}
& d \text{ even}  \\
\sum_{j=\frac{d-1}{2}}^{\frac{d+1}{2}}\binom{d}{j} q^{d-j} \tilde{q}^{j}
& d \text{ odd}  
\end{cases}
\\ =
&\batta(\Ldens{c})\begin{cases}
 \binom{d}{\frac{d}{2}} (p\omp)^{\frac{d}{2}} 
& d \text{ even}  \\
\binom{d}{\frac{d-1}{2}} (p\omp)^{\frac{d-1}{2}} 
\batta(D(p,\Kcheck))
& d \text{ odd}  
\end{cases}
\\ \le
&\batta(\Ldens{c})\begin{cases}
 (4 p)^{\frac{d}{2}} 
& d \text{ even}  \\
 (4 p)^{\frac{d-1}{2}} 
\batta(D(p,\Kcheck))
& d \text{ odd}  
\end{cases}
\end{align*}
To get the final bound on $\gamma'\batta(\Ldens{m}')$ we need to multiply the above bounds by $d\bar{\gamma}\gamma^{d-1}$ when $\ntm = 1$ and by $\gamma^d$ when $\ntm = 0$.
In the next section we will use $\batta(D(p,\Kcheck)) \leq \batta(\Ldens{b})$ to further bound the above
expressions.

\subsection{Stability with Minimum Degree $3$.}

Let us assume that the minimum variable node degree, given by $d+1$, is at least three
and a right regular degree $d_r+1.$ 

In view of \eqref{vnodeallbnd} and \eqref{checkallbnd} we may assume
$\batta{(\Ldens{a}^{(n)})}  \le 3e^{-\frac{\Kvar}{2}}$  which implies
$\batta{(\Ldens{b}^{(n)})}  \le  3 d_r e^{-\frac{\Kvar}{2}}$, $\gamma^{(n)} p^{(n)}  e^{\frac{\Kvar}{2}} \le 3e^{-\frac{\Kvar}{2}}$ and $\batta(\Ldens{m}^{(n)})\leq  3e^{-\frac{\Kvar}{2}}$
for all $n \ge N$ for some $N\in \naturals.$ Here we use the notation,
$\Ldens{a}^{(n)} = \gamma^{(n)} D(p^{(n)},\Kvar) +\bar{\gamma}^{(n)}\Ldens{m}^{(n)}$.
We assume $\Kvar$ large enough so that for all $d$ we have
\[
 \frac{d(d-1)}{2} \batta(\Ldens{c})\batta(\Ldens{b}^{(n)})^{d-2} \le 1.
\]
We put together everything done previously to bound the contributions to the density coming out of the variable nodes at the $(n+1)$th iteration. To do this, we first use the check node analysis in Lemma~\ref{lem:checknodeanalysis} with incoming density given by $\Ldens{a}^{(n)}$. Then, using the variable node analysis  of the previous section we obtain
\begin{align}
\begin{split}\label{eqn:gpitupdate}
\gamma^{(n\!+\!1)} p^{(n\!+\!1)}  \le &
e^{-\frac{\Kvar}{2}}  (d_r \xi \bar{\gamma}^{(n)}\batta(\Ldens{m}^{(n)}))^2 
\,\\+
(d+&1) (4 d_r  \gamma^{(n)} p^{(n)})^{\lfloor\frac{d}{2}\rfloor+1}\,\\+
d (4 &d_r  \gamma^{(n)} p^{(n)})^{\lfloor\frac{d}{2}\rfloor}  \batta(\Ldens{c})d_r\xi (\bar{\gamma}^{(n)}\batta(\Ldens{m}^{(n)})),
\end{split}
\\
\begin{split}\label{eqn:mitupdate}
 \bar{\gamma}^{(n+1)}\batta(\Ldens{m}^{(n+1)})  \le &
 (d_r \xi \bar{\gamma}^{(n)} \batta(\Ldens{m}^{(n)}))^2  
\\ + 
2d\batta(\Ldens{c})  d_r \xi  (\bar{\gamma}^{(n)}& \batta(\Ldens{m}^{(n)}) (d_r 4\gamma^{(n)} p^{(n)})^{\lfloor\frac{d-2}{2}\rfloor}\batta(\Ldens{b}^{(n)})
\\ +
\batta(\Ldens{c})  (d_r 4\gamma^{(n)} & p^{(n)})^{\lfloor\frac{d}{2}\rfloor}
\end{split}
\end{align}
To obtain the second inequality we use $\batta(\Ldens{b}^{(n)}) \leq 1$, where we assume $\Kvar$ large enough so that $3 d_r e^{-\frac{\Kvar}2} \leq 1$.

Now for any $\epsilon > 0$ we choose $\Kvar$ large enough so that 
$(d_r 4\gamma^{(n)} p^{(n)}) < 1$
and for all $d \ge 2$
we have
\begin{align*}
 \epsilon \ge &
(d_r \xi)^2 \bar{\gamma}^{(n)}\batta(\Ldens{m}^{(n)})  
+
2d\batta(\Ldens{c})d_r \xi  \batta(\Ldens{b}^{(n)}),
\\
 \epsilon \ge &
e^{-\frac{\Kvar}{2}}\batta(\Ldens{c})  4 d_r, 
\\
 \epsilon \ge &
e^{-\frac{\Kvar}{2}}4 d_r (d+1) \Bigl(1+
  \batta(\Ldens{c})d_r\xi (\bar{\gamma}^{(n)}\batta(\Ldens{m}^{(n)}))\Bigr),
\end{align*}
which then yields
\begin{align}\label{eqn:epsineq}
\begin{bmatrix}    \bar{\gamma} \batta{(\Ldens{m})}
\\ e^{\frac{\Kvar}{2}} \gamma    p   \end{bmatrix}^{(n+1)}
\le
\epsilon
\begin{bmatrix} 1 & 1 \\ 1& 1 \end{bmatrix}
\begin{bmatrix} \bar{\gamma} \batta{(\Ldens{m})}  \\
 e^{\frac{\Kvar}{2}}  \gamma    p  
\end{bmatrix}^{(n)},
\end{align}

%
%
where $[\cdot]^{(n)}$ denotes the values at the $n$th iteration.
We summarize our findings in the following.
\begin{theorem}
Consider an irregular ensemble with check regular degree $\dr$ and minimum
variable node degree at least three. If a channel $\Ldens{c}$ is below the
BP threshold then it is below the threshold for SatBP for $\Kvar$ sufficiently large.
\end{theorem}
\begin{IEEEproof}
Assume the channel $\Ldens{c}$ is below the BP threshold.
Let $x^*$ be the constant of Lemma \ref{lem:nearstability}.
Under BP we have $\batta(T^{(\ell)}(\Ldens{c}, \Delta_0)) < x^*/2$ 
for some $\ell$ large enough.
By Lemma \ref{lem:clippeddecodervssymclipping}
and Lemma \ref{lem:distsymclippedandBP}
we have $\batta(S^{(\ell)}_\Kvar(\Ldens{c}, \Delta_0)) \le x^*$ for $\Kvar$ large enough.
By Lemma \ref{lem:nearstability}, and assuming $\Kvar$ large enough,
 we have $\batta(S^{(n)}_\Kvar(\Ldens{c}, \Delta_0))\le 3e^{-\frac{\Kvar}{2}}$ for all $n$ large enough.
The stability analysis above then implies that
$\lim_{n\rightarrow \infty}\perr{(S^{(n)}_\Kvar(\Ldens{c}, \Delta_0))} =0.$
\end{IEEEproof}

\section{Block Thresholds and Speed of Convergence\label{sec:Speed}}
Thresholds for iterative coding systems are usually {\em bit} thresholds.
In some cases one can show that the iterative block error rate has the
same threshold \cite{1523291,1522644}.
For standard irregular ensembles it is sufficient that variable node 
degrees are at least three.  The key observation for degree  three and above is that below the bit threshold the bit error rate converges to zero doubly exponentially in iteration.  One can maintain tree-like neighborhoods
with blocklength growing exponentially in iteration and therefore the
block error rate can be shown to converge to zero.
In \cite{1523291} it was shown that degree two
variable nodes connected in an accumulate structure could be admitted
while retaining the block threshold result provided an appropriate update schedule was adopted.  
The key idea there was that, by effectively updating a string of degree two updates in sequence
for each iteration, one could achieve exponential decay in error probability with as large 
and exponent as required.

In this section we consider the impact of saturation on the block threshold.
The stability analysis for ensembles with minimum variable node degree  three
shows exponential decay in iteration of bit error probability with arbitrarily large exponent.
Consequently, we can show for a suitable ensemble that the block threshold coincides with
the bit threshold.
Nevertheless, saturation has a pronounced effect on stability and we observe this especially in 
the conditions required for doubly exponential convergence of the bit error probability.
We show that doubly exponential convergence occurs for SatBP with minimum variable node
degree five.  With minimum variable node degree four doubly exponential convergence does not
occur but
can be recovered the addition of a single extra LLR magnitude and a two-tiered saturation.
For minimum variable degree three doubly exponential convergence of the bit error rate
can be recovered with a more radical modification of the decoding process 
(erase received values once the bit error rate is sufficiently small.)

Let us briefly review the standard block threshold arguments.
For further details we refer to \cite{1523291,1522644}.
Density evolution gives the bit error rate $P_b(\ell)$ as a function of iteration assuming 
tree-like neighborhoods up to iteration $\ell.$  For block length $n$ the block error rate,
assuming tree-like neighborhoods, is upper bounded by $n P_b(\ell).$
For the block error rate analysis we require that {\em all} computation trees are tree-like.
This is accomplished through an expurgation or modification of the standard ensemble.
The simplest approach, and the one we adopt, is to consider $n = n(\ell)$ large enough
so that the fraction of variable nodes whose neighborhoods are not tree-like tends to zero
as $\ell$ gets large.  Then, we modify the code by declaring the associated bits as known and set to $0.$
This lowers slightly the rate of the code and in effect modifies slightly the degree structure.
The net effect is an improvement in bitwise performance.
Asymptotically in large $\ell$ the modification is negligible so that full rate is recovered.

The basic calculation is as follows.
Consider a computation tree associated to $\ell$ iterations.
Let ${\cal M}_\ell$ denote the number of variable nodes in the computation tree.
Let $n \gg {\cal M}_\ell$ denote the block length.
It is not difficult to see that there exists a constant $\gamma$ independent of $\ell$ and $n$
such that the probability that the neighborhood is tree like is at least
\[
 (1- \gamma\frac{{\cal M}_\ell}{n})^{{\cal M}_\ell} \ge  (1- \gamma\frac{{\cal M}_\ell^2}{n})
\]
Now, we have a bound of the form ${\cal M}_\ell^2 \le e^{M \ell}$ (where $M$ depends on the
degree structure) and we choose
$n   =e^{N \ell}$ where $N > M.$ Thus $N$ depends only on the degree structure of the code.  It then follows that the fraction of variable nodes whose
neighborhoods are not tree-like is tending to $0$ in $\ell.$  To show that the block threshold
equals the bit threshold it remains only to show that
\[ \lim_{\ell \rightarrow \infty} e^{N  \ell} P_b(\ell)  = 0. \]
It is sufficient therefore to show that
\[ \liminf_{\ell \rightarrow \infty} (-\ln P_b(\ell))  > N\,. \]

Let us consider 
\(
E(\ell) : = \begin{bmatrix} 1 & 1
\end{bmatrix}
\begin{bmatrix} \bar{\gamma} \batta{(\Ldens{m})}  \\
 e^{\frac{\Kvar}{2}}  \gamma    p  
\end{bmatrix}^{(\ell)}\,.
\)
We clearly have 
$$P_b(\ell) \le E(\ell) =\bar{\gamma}^{(\ell)} \batta{(\Ldens{m}^{(\ell)})}  +  e^{\frac{\Kvar}{2}}  \gamma^{(\ell)}    p^{(\ell)}.$$

From the previous analysis we know that there exists an $\ell_0$ such that $E(\ell_0)$ is small. Recursing equation \eqref{eqn:epsineq}, we get $E(\ell+\ell_0) \leq (2\epsilon)^{\ell}E(\ell_0) = E(\ell_0)e^{-\ell\ln (1/(2\epsilon))}$. 
We can now make $\epsilon$ arbitrarily small by choosing $\Kvar$ large enough.
Hence for sufficiently large $\Kvar$ we obtain 
\[ \liminf_{\ell \rightarrow \infty} (-\ln E(\ell)) > N\, \]
thus establishing the desired result.

\subsection{Variable nodes with Minimum Degree at least 5}
In this section we show that SatBP does achieve doubly exponential convergence in $\ell$
of the error probability when the variable node degrees are at least five.

The rate of convergence depends largely on the variable node update.
It is clear from \eqref{eqn:gppbound} that, even with degree three,
$\gamma' p'$ has quadratic dependence on $\gamma p$ and $\bar{\gamma} \batta(\Ldens{m}).$
For doubly exponential convergence we can admit linear dependence of
$\bar{\gamma}\batta(\Ldens{m})$ on $\gamma p,$ but the dependence on 
$\bar{\gamma} \batta(\Ldens{m})$ must be of higher order. Let us make this more precise.

As before we assume $\Kvar$ and $N$ large enough so that for all $d$ and
$n \ge N$ we have
\(
 \frac{d(d-1)}{2} \batta(\Ldens{c})\batta(\Ldens{b}^{(n)})^{d-2} \le 1
\)
and $(4 d_r \gamma^{(n)} p^{(n)}) < 1$. Then from 
\eqref{eqn:gpitupdate} and \eqref{eqn:mitupdate}, assuming $d \ge 4,$ we get 
\begin{align}
\begin{split}\label{eqn:gpitupdateB}
 e^{\frac{\Kvar}{2}} \gamma^{(n+1)} p^{(n+1)}  & \le
  (d_r \xi \bar{\gamma}^{(n)}\batta(\Ldens{m}^{(n)}))^2 
\\  +
e^{-{\Kvar}} (4 & d_r   e^{\frac{\Kvar}{2}}\gamma^{(n)} p^{(n)})^{3}\,\\
+
e^{-\frac{\Kvar}{2}} (4 & d_r  e^{\frac{\Kvar}{2}}\gamma^{(n)} p^{(n)})^{2}  \batta(\Ldens{c})d_r\xi (\bar{\gamma}^{(n)}\batta(\Ldens{m}^{(n)})),
\end{split}
\\
\begin{split}\label{eqn:mitupdateB}
\bar{\gamma}^{(n+1)} \batta(\Ldens{m})^{(n+1)} & \le  
 (d_r \xi \bar{\gamma}^{(n)}\batta(\Ldens{m}^{(n)}))^2  
\\+
 2d\batta(\Ldens{c})d_r \xi  (\bar{\gamma}^{(n)}&  \batta(\Ldens{m}^{(n)})  e^{-\frac{\Kvar}{2}}(d_r 4 e^{\frac{\Kvar}{2}}\gamma^{(n)} p^{(n)}) \batta(\Ldens{b}^{(n)})
\\ +
 e^{-{\Kvar}}\batta(\Ldens{c}) (d_r   4 & e^{\frac{\Kvar}{2}}  \gamma^{(n)} p^{(n)})^{2},
\end{split}
\end{align}
from which we easily obtain that for $\Kvar$ large enough we have
\begin{align*}
\bar{\gamma}^{(n+1)} \batta(\Ldens{m})^{(n+1)}
& +
e^{\frac{\Kvar}{2}}\gamma^{(n+1)} p^{(n+1)} 
 \le \\
 & 2 (d_r \xi)^2 (\bar{\gamma}^{(n)} \batta(\Ldens{m})^{(n)}
+
e^{\frac{\Kvar}{2}}\gamma^{(n)} p^{(n)})^2,
\end{align*}
which yields doubly exponential convergence in the iterations.


\subsection{Decoder Alteration for Degree Four}

When $d=3$ (degree four) the SatBP decoder does not yield doubly exponential 
stability convergence.
The limiting effect arises in the variable node analysis
from messages of type 
$(\nti = 0,\ntm = 1,\ntj = 2)$ which contribute a linear dependence of 
$\batta(\Ldens{m'})$ on $\batta(\Ldens{m}).$
This occurs because $0 < 2 \Kcheck - (\Kcheck+\Kchannel) < \Kvar.$
If the support of $\Ldens{m}$ were reduced to $[-\Kalt,\Kalt]$ where
$2 \Kcheck - (\Kalt+\Kchannel) > \Kvar$ then this term would be eliminated and 
doubly exponential convergence can be recovered.  

Thus, for minimum degree four we consider a two step saturation at variable nodes
where all messages with magnitude at least $\Kvar$ are saturated to $\Kvar$
and messages with magnitude between $\Kalt$ and $\Kvar$ are saturated to $\Kalt.$
Hence, for this section we assume
the inequality
\[
2\Kcheck-\Kvar  \ge \Kchannel +\Kalt\,.
\]
We assume $\lambda \in (\frac{1}{2},1]$ and note that the above inequality then implies
$\Kchannel \le (1-\lambda) \Kvar.$

Note that an equivalent interpretation under scaling of the saturation levels is
that we append an additional magnitude level to the SatBP decoder. 
Under this interpretation we identify $\Kalt$ with $\Kvar$ and $\Kvar$ with
$\lambda^{-1} \Kvar$ where magnitudes above this level are saturated to
$\lambda^{-1} \Kvar.$  Under this interpretation the modification appears as an 
improvement on SatBP and, using this perspective, it is relatively easy to reproduce
the results on the approximation of BP by the saturating decoder.  Let us make this more precise. 
For notational purposes we will adhere to the original interpretation.

Let $\altSatLdens{a}_{\lambda,\Kgen}$ denote the double saturation of $\Ldens{a}$ and
let $\altSatLdens{a}_\sym{\lambda,\Kgen}$ denote the symmetrized version.
Let $S_{\lambda,\sym{\Kgen}}$ denote the corresponding one step density evolution update.
We easily obtain the following generalization of 
Lemma \ref{lem:DistSymClip}
\begin{align*}
d(\Ldens{a}, \altSatLdens{a}_{\lambda,\sym{\Kgen}}) 
\leq d(\Ldens{a}, \SatLdens{a}_{\lambda\sym{\Kgen}}) 
\leq 1 - \tanh(\lambda{\Kgen}/2),
\end{align*}
where $\Ldens{a}$ is any symmetric $L$-density.
It is not hard to see that we can also obtian the following generalization of Lemma \ref{lem:distsymclippedandBP},
\begin{align*}
d(T^{(\ell)}&(\Ldens{c}, \Delta_0), S_{\lambda,\sym{\Kgen}}^{(\ell)}(\Ldens{c}, \Delta_0)) \leq 2e^{-\lambda \Kgen + \ell\cdot\ln (2(\dl-1)(\dr-1))}.
\end{align*}
The relationship between the symmetrized decoder and the non-symmetrized version as analyzed in 
in Lemma \ref{lem:clippeddecodervssymclipping} remains essentially unchanged and we have that
for any $0 < \epsilon < 1$ and $\ell \in \naturals$, there exists a $\Kvar$ large enough
such that
$$
\batta(S^{(\ell)}_{\lambda,\Kgen}(\Ldens{c}, \Delta_0)) \leq \frac1{1-\epsilon}\batta(S^{(\ell)}_{\lambda,\sym{\Kgen}}(\Ldens{c}, \Delta_0)).
$$

We can now focus our attention on the stability analysis.
Let $\Ldens{a}$ be a density supported on $[-\Kvar,\Kvar].$
Then we have the two bounds,
\begin{align}
\batta(\altSatLdens{a}_{\lambda,\Kvar}) & \le  e^{\frac{\Kvar-\Kalt}{2}} \batta( \SatLdens{a}_{\Kvar}),
\label{eq:modvarnodeineq}\\
\batta(\altSatLdens{a}_{\lambda,\Kvar})  &\le   \batta (\Ldens{a})+e^{-\frac{\Kalt}{2}}\,.
\end{align}
The first (multiplicative) inequality is new and will be used to establish doubly exponential convergence. Indeed, 
since $e^{\frac{\Kvar - \Kalt}2} \geq 1$, we have 
\begin{align*}
& \batta(\altSatLdens{a}_{\lambda,\Kvar})  \leq  e^{\frac{\Kvar - \Kalt}2}e^{\frac{\Kvar}2}\int_{-\infty}^{-\Kvar}\Ldens{a}(x)dx \\ &  + \int_{-\Kvar}^{\Kalt}\!\!\!\!e^{-\frac{x}2}\Ldens{a}(x)dx  
  + e^{\frac{\Kvar - \Kalt}2}\int_{\Kalt}^{\Kvar} e^{-\frac{x}2}\Ldens{a}(x)dx \\
&  + e^{\frac{\Kvar - \Kalt}2} e^{-\frac{\Kvar}2}\int_{\Kvar}^{\infty}\Ldens{a}(x)dx \leq e^{\frac{\Kvar-\Kalt}{2}} \batta( \SatLdens{a}_{\Kvar}).
\end{align*}
The second (additive) inequality allows us to reproduce the near stability analysis of Section \ref{sec:nearstab}
to obtain
as in the derivation of \ref{vnodeallbnd} and \ref{checkallbnd} for the doubly saturated decoder
the bounds
\begin{align}
\batta{(\Ldens{a}^{(n)})} & \le 3e^{-{\Kalt}/2} \label{vnodeallbndA}\\
\batta{(\Ldens{b}^{(n)})} & \le  3\rho'(1) e^{-{{\Kalt}/2}}\label{checkallbndA}\,.
\end{align}
which hold for $n \ge N$ (for some $N\in\naturals$) and $\Kvar$ large enough assuming the channel is below the
BP threshold. 

We assume that no additional saturation is performed at the check node
so, in particular, Lemma \ref{lem:checknodeanalysis} still applies.
In the variable node analysis we note that \eqref{eqn:gppbound} still applies.
The change in the analysis concerns the bound on $\bar{\gamma}'\Ldens{m'}$
in the variable node analysis.
New considerations apply to the inner saturation of the density $\Ldens{m'}.$ 
Further note that the incoming densities in to the variable nodes have support on $\pm\Kcheck\cup (-\Kalt,\Kalt)$.
First we note the contribution from types with $\ntm \ge  2.$ Let the notation $\altSatLdens{a}_{\lambda, \Kvar}^\circ$ denote the density on the support $[-\Kalt, \Kalt]$ which is equivalent, in this case, to the support on $(-\Kvar, \Kvar)$. Using analysis in the previous section and the inequality \eqref{eq:modvarnodeineq} we get,
\begin{align*}
&\batta\left( \altSatLdens{\sum_{k=0}^{d-2} \binom{d}{k}\gamma^{k}\bar{\gamma}^{d-k}\Ldens{c}\vconv \Dd^{\vconv k}\vconv \Ldens{m}^{\vconv (d-k)}}_{\lambda, \Kvar}^\circ \right)
\le
\\&
 e^{\frac{\Kvar-\Kalt}{2}}  \frac{d(d-1)}{2} (\bar{\gamma}\batta(\Ldens{m}))^2 \batta(\Ldens{c})\batta(\Ldens{b})^{d-2}.
\end{align*}

Now we consider the contribution from types with $\ntm=1.$
A type $(\nti,1,\ntj)$ can have a non-zero contribution to $\Ldens{m'}$ only if the interval centered on
$(\ntj-\nti)\Kcheck$ of width $2(\Kchannel+\Kalt)$ intersects $(-\Kvar,\Kvar).$
Since we assume $2\Kcheck  \ge \Kchannel+\Kvar +\Kalt$ and $\Kcheck \le \Kvar$ we obtain
\begin{align*}
&\batta(\SatLdens{\Ldens{c} \vconv \Ldens{m} \vconv \Dd^{d-1}}_{\Kvar}^\circ)
\le
\\&
\batta(\Ldens{c})\batta(\Ldens{m})
\begin{cases}
\sum_{j=\frac{d-2}{2}}^{\frac{d}{2}} \binom{d-1}{j} q^{d-1-j} \tilde{q}^j
& d \text{ even}  \\
 \binom{d-1}{\frac{d-1}{2}} q^{\frac{d-1}{2}} \tilde{q}^\frac{d-1}{2}
& d \text{ odd},  
\end{cases}
\end{align*}
where recall that  $q = e^{\frac{\Kcheck}{2}}p$ and $\tilde{q} =e^{-\frac{\Kcheck}{2}}\omp.$
Again, combining the above with \eqref{eq:modvarnodeineq}, we obtain
\begin{align*}
&\batta\left(\altSatLdens{\Ldens{c} \vconv \Ldens{m} \vconv \Dd^{d-1}}_{\lambda,\Kvar}^\circ\right)\le
\\  &
 e^{\frac{\Kvar-\Kalt}{2}}\batta(\Ldens{c})\batta(\Ldens{m})
\begin{cases}
\binom{d-1}{\frac{d}{2}} (p\omp)^{\frac{d-2}{2}}\batta(D(p,\Kcheck))
& d \text{ even}  \\
 \binom{d-1}{\frac{d-1}{2}} (p\omp)^{\frac{d-1}{2}}  
& d \text{ odd}. 
\end{cases}
\end{align*}

Finally we consider the contribution from types with $\ntm=0.$
A type $(\nti,0,\ntj)$ will have a non-zero contribution to $\Ldens{m'}$ only if the interval centered on
$(\ntj-\nti)\Kcheck$ of width $2\Kchannel$ intersects $(-\Kvar,\Kvar).$
Hence we obtain
\begin{align*}
& \batta(\SatLdens{\Ldens{c}  \vconv \Dd^{d}}_{\Kvar}^\circ)
 \le
\batta(\Ldens{c})\begin{cases}
 \binom{d}{\frac{d}{2}} q^{\frac{d}{2}} \tilde{q}^{\frac{d}{2}}
& d \text{ even}  \\
\sum_{j=\frac{d-1}{2}}^{\frac{d+1}{2}}\binom{d}{j} q^{d-j} \tilde{q}^{j}
& d \text{ odd}  
\end{cases}
\end{align*}
which gives
\begin{align*}
& \batta(\altSatLdens{\Ldens{c}  \vconv \Dd^{d}}_{\lambda, \Kvar}^\circ)
\\ \le&
\batta(\Ldens{c})\begin{cases}
 \binom{d}{\frac{d}{2}} (p\omp)^{\frac{d}{2}} 
& d \text{ even}  \\
 e^{\frac{\Kvar-\Kalt}{2}}\binom{d}{\frac{d-1}{2}} (p\omp)^{\frac{d-1}{2}} 
\batta(D(p,\Kcheck))
& d \text{ odd}  
\end{cases}
\end{align*}

Since $\lambda > \frac{1}{2}$
we can assume for $d \ge 3$ and for $\Kvar$ large enough that,
\begin{align*}
 e^{\frac{\Kvar-\Kalt}{2}} & \frac{d(d-1)}{2}  \batta(\Ldens{c})\batta(\Ldens{b})^{d-2}
 \\ \stackrel{\eqref{checkallbndA}}{\le} & e^{-\frac{2\lambda-1}{2}\Kvar}\frac{d(d-1)}{2}  3\rho'(1)\batta(\Ldens{c})\batta(\Ldens{b})^{d-3}
\\ \le & 1.
\end{align*}
Also,
\begin{align*}
&\batta(\altSatLdens{\Ldens{c} \vconv \Ldens{m} \vconv \Dd^{d-1}}_{\lambda,\Kvar}^\circ)\le
\\  &
 e^{\frac{\Kvar-\Kalt}{2}}\batta(\Ldens{c})\batta(\Ldens{m})
\begin{cases}
\binom{d-1}{\frac{d}{2}} (p\omp)^{\frac{d-2}{2}}\batta(D(p,\Kcheck))
& d \text{ even}  \\
 \binom{d-1}{\frac{d-1}{2}} (p\omp)^{\frac{d-1}{2}}  
& d \text{ odd}  
\end{cases}
\\& \leq e^{\frac{\Kvar-\Kalt}{2}}\batta(\Ldens{c})\batta(\Ldens{m})
\begin{cases}
 (4p)^{\frac{d-2}{2}}
& d \text{ even}  \\
(4p)^{\frac{d-1}{2}}  
& d \text{ odd}. 
\end{cases}
\end{align*}
Finally,
\begin{align*}
& \batta(\altSatLdens{\Ldens{c}  \vconv \Dd^{d}}_{\lambda,\Kvar}^\circ)
 \le
\batta(\Ldens{c})\begin{cases}
  (4p\omp)^{\frac{d}{2}} 
& d \text{ even}  \\
 e^{\!-\!\frac{2\lambda\!-\!1}{2}\Kvar} \!3\rho'(1)(4p)^{\frac{d\!-\!1}{2}} 
& d \text{ odd}  
\end{cases}
\\ & \quad\quad\quad\quad\quad\quad\quad\quad\le
(4p)^{\lfloor\frac{d-1}{2}\rfloor}.
\end{align*}
To get the final bound on $\gamma'\batta(\Ldens{m}')$ we need to multiply the above bounds by $d\bar{\gamma}\gamma^{d-1}$ when $\ntm = 1$ and by $\gamma^d$ when $\ntm = 0$. For $\Kvar$ large enough we can make $4\gamma p \leq 1$. Thus we get,
\begin{align*}
\bar{\gamma}'\batta(\Ldens{m}') \le
\batta(\Ldens{c})&
\Bigl(
(\bar{\gamma}\batta(\Ldens{m}))^2
\!+\!
 de^{\frac{\Kvar\!-\!\Kalt}{2}}\batta(\Ldens{c})(\bar{\gamma}\batta(\Ldens{m})) (4 \gamma p)
\\ & +
(4  \gamma p)
\Bigr)
\end{align*}
Assuming $d \ge 3$ we also have from the previous analysis,
\begin{align}
\begin{split}
\gamma p' \le
&e^{-\frac{\Kvar}{2}} \frac{d(d-1)}{2} (\bar{\gamma}\batta(\Ldens{m}))^2 \batta(\Ldens{c})\batta(\Ldens{b})^{d-2}
\\&+
 (d+1)(4 \gamma p)^{\lfloor\frac{d}{2}\rfloor+1}\,+
d(4 \gamma p)^{\lfloor\frac{d}{2}\rfloor}  \batta(\Ldens{c})(\bar{\gamma}\batta(\Ldens{m})).
\end{split}\nonumber
\end{align}

Thus we now obtain quadratic dependence and hence doubly exponential convergence even when minimum variable node degree is four.

\subsection{Decoder Alteration for Degree Three}

In this section we will show that when the minimum variable node degree is 3, we can still have doubly exponential convergence of the bit error rate which implies an exponential (in blocklength) convergence of the block error rate with a decoder alteration. In this case, however, we require an iteration dependent alteration of the decoder.
We alter the decoder only after the error rate is sufficiently small. Hence, for the analysis 
we  assume operation in the near stability region. 
More precisely, we have $\batta{(\Ldens{a})} \le 3e^{-{{\Kvar}/2}}$, where $\Ldens{a}$ is the outgoing density at the variable nodes. Since $\Ldens{a} = \gamma D(p,\Kvar) +\bar{\gamma}\Ldens{m}$, we further have $\bar{\gamma}\batta{(\Ldens{m})} \le 3e^{-{{\Kvar}/2}}$ and  $\gamma p  \le 3e^{-{{\Kvar}/2}}.$ 

We note that the previous technique of saturation at two levels does not yield the quadratic dependence we seek for the term $\batta(\Ldens{m'})$. Indeed, any incoming density having the type $(\nti = 0,\ntm = 1,\ntj = 1)$ will always  contribute to the outgoing density of type $\Ldens{m'}$, implying linear dependence of $\batta(\Ldens{m'})$ on $\batta(\Ldens{m}).$
To show doubly exponentially fast convergence of the bit error rate, we modify the decoder as follows. After the messages have become reasonably good, i.e., we are in the near stability region, we erase the channel information. The intuition is that at this point the extrinsic information is good enough for successful decoding. Then for every incoming message we make a hard-decision to either $+1$ or $-1$ based on the sign of its LLR value. The decoding algorithm then proceeds in a manner similar to the erasure decoder \cite{RiU08}. Let us explain this in more detail.   

The decoder has now three messages $\{-1,0,+1\}$. At the variable node side, there is an erasure message on the outgoing edge if and only if all the incoming messages are erasures or there is exactly one $+1$ and $-1$ message. The outgoing edge carries a $-1$ message if and only if all incoming messages are $-1$ or one message is an erasure and the other is $-1$. At the check node side, the outgoing message is an erasure if at least one incoming message is an erasure, else the outgoing message is the product of the incoming messages.  We can now write the density evolution equation analysis for this decoder as follows. Let $x_\ell$ and  $y_\ell$ represent the probability of the messages $0$ and $-1$, respectively, coming out of the variable node. Also, let $w_\ell$ and  $z_\ell$ represent the probability of the messages $0$ and $-1$, coming out of the check node respectively. Since we are in the near stability region, it is not hard to see that $x_0 \leq \bar{\gamma}\batta{(\Ldens{m})} \leq ce^{-{\Kvar}/2}$ and $y_0 \leq \batta{(\Ldens{a})} \leq ce^{-{\Kvar}/2}$. Indeed, $y_0 = \int_{x<0} \Ldens{a}(x)dx \leq \int_{x\leq 0} \Ldens{a}(x)e^{-x/2}dx \leq \batta(\Ldens{a}).$ From the decoder rules we immediately get,
\begin{align*}
x_\ell \stackrel{(a)}{\leq} & w^2_{\ell} + z_{\ell}, \\
y_\ell = & z^2_{\ell} + w_{\ell}z_{\ell}, \\
w_{\ell} = & 1 - (1 - x_{\ell-1})^{d_r-1} \leq (d_r - 1) x_{\ell-1}, \\
z_{\ell} \stackrel{(b)}{\leq} & 1 - (1 - y_{\ell-1})^{d_r-1} \leq (d_r - 1) y_{\ell-1},
\end{align*}
where $d_{r}$ is the check node degree. To obtain $(a)$ we simply upper bound the probability of message with value $+1$ by 1. At the check node side, the outgoing message is $-1$ if there are odd number of incoming messages that are $-1$. This implies that at least one incoming message must be $-1$ and hence we obtain inequality $(b)$.

Combining the four inequalities above, it is not hard to see that $x_{\ell} + y_{\ell} \leq C(x_{\ell-2} + y_{\ell-2})^2$ for some positive constant $C$. This implies $x_{\ell} + y_{\ell} \leq (Ax_0)^{2^{n/2}}$, where $A$ is some positive constant and $n$ is the number of iterations of the erasure decoder. Hence we obtain the doubly exponential convergence.

\section{Threshold for the SatBP Decoder and Channels with Infinite Support}



Consider a channel family, BMS($\ent$), ordered by $\ent$ and let $\ent^{\BPsmall}(\lambda, \rho)$ denote the BP threshold when transmitting over this channel family using a $(\lambda, \rho)$ ensemble. Also, a priori the channel has support on $(-\infty,\infty)$.

Let us describe the analysis of the SatBP decoder in this case. Consider transmission over a channel with $L$-density $\Ldens{c}$. From the previous analysis we have that the channel support must be finite for stability of the perfect decoding fixed point when we use the SatBP decoder. As a result, we saturate the channel $\Ldens{c}$ to a value $\Kchannel\leq 2\Kcheck- \Kvar$ before we feed it to the SatBP decoder. The value $\Kcheck$ is defined in section~\ref{sec:stabmindegree3}.  Thus we consider transmission over a channel $\lfloor\Ldens{c}\rfloor_{\Kchannel}$. 

For the purpose of analysis we also consider the corresponding symmetric channel, achieved via flipping as explained previously. Denote it by $\lfloor\Ldens{c}\rfloor_\sym{\Kchannel}.$ We have the following lemma.
\begin{lemma}[Stability Condition for Sym. Sat. Channels]
Consider transmission over a general BMS channel $\Ldens{c}$ using $(\lambda,\rho)$ ensemble. Let $\Ldens{c} \in$ BMS($\ent$) be such that it satisfies the following stability condition,
\begin{align*}
(\lambda'(0)\rho'(1))(\batta(\Ldens{c})  + 2e^{-\Kchannel/2}) < 1.
\end{align*}
Then, the full BP decoder is successful when transmitting over the symmetric channel $\lfloor\Ldens{c}\rfloor_\sym{\Kchannel}$. Furthermore, the loss in capacity is also bounded by
$ \frac{2}{\ln 2}e^{-\Kchannel/2}$.
\end{lemma}
\begin{proof}
We bound the Wasserstein distance between the DE with channel $\Ldens{c}$ and DE with channel $\lfloor\Ldens{c}\rfloor_\sym{\Kchannel}$ as follows,
\begin{align*}
& d(T^{(\ell)}(\Ldens{c},\Delta_0),T^{(\ell)}(\lfloor\Ldens{c}\rfloor_\sym{\Kchannel},\Delta_0)) = \\
& d(T(\Ldens{c},T^{(\ell-1)}(\Ldens{c},\Delta_0)),T(\Ldens{c},T^{(\ell-1)}(\lfloor\Ldens{c}\rfloor_\sym{\Kchannel},\Delta_0))) + \\
& d(T(\Ldens{c},T^{(\ell-1)}(\lfloor\Ldens{c}\rfloor_\sym{\Kchannel},\Delta_0)),T(\lfloor\Ldens{c}\rfloor_\sym{\Kchannel},T^{(\ell-1)}(\lfloor\Ldens{c}\rfloor_\sym{\Kchannel},\Delta_0))) \\
& \stackrel{\text{(vi,viii), Lem. 13 in \cite{KRU12b}}}{\leq} \alpha_\ell d(T^{(\ell-1)}(\Ldens{c},\Delta_0),T^{(\ell-1)}(\lfloor\Ldens{c}\rfloor_\sym{\Kchannel},\Delta_0)) \\
& \quad\quad\quad\quad\quad +  2d(\Ldens{c},\lfloor\Ldens{c}\rfloor_\sym{\Kchannel}) \\
& = \alpha_\ell d(T^{(\ell\!-\!1)}(\Ldens{c},\!\Delta_0),T^{(\ell\!-\!1)}(\lfloor\Ldens{c}\rfloor_\sym{\Kchannel},\!\Delta_0)) \!+\!  2(1\!-\! \tanh(\frac{\Kchannel}2)),
\end{align*}
where
\begin{align*}
\alpha_{\ell} & = 2 (\dl-1)  
 \sum_{j\!=\!1}^{\dr\!-\!1} (1\!-\!\batta^2(\Ldens{a}))^{\frac{\dr\!-\!1\!-\!j}2}(1\!-\!\batta^2(\Ldens{b}))^{\frac{j\!-\!1}2},
\end{align*}
where $\Ldens{a} = T^{(\ell - 1)}(\Ldens{c}, \Delta_0)$ and $\Ldens{b} = T^{(\ell - 1)}(\lfloor\Ldens{c}\rfloor_\sym{\Kchannel}, \Delta_0)$ and $\dl$ and $\dr$ correspond to the average variable node and check node degrees. Following the same steps as in the proof of lemma~\ref{lem:distsymclippedandBP} we get
\begin{align*}
\batta(T^{(\ell)}&(\lfloor\Ldens{c}\rfloor_\sym{\Kchannel}, \Delta_0))  \\ 
& \leq \batta(T^{(\ell)}(\Ldens{c}, \Delta_0)) + 2\sqrt{2}e^{\frac{-\Kchannel + \ell\cdot\ln (2(\dl-1)(\dr-1))}2}.
\end{align*}
Thus, for any $\xi > 0$, we can choose $\Kchannel$ large enough, such that $\batta(T^{(\ell)}(\lfloor\Ldens{c}\rfloor_\sym{\Kchannel}, \Delta_0)) \leq \xi$ for all $\ell \geq \ell_0$. Here $\ell_0$ is such that $\batta(T^{(\ell_0)}(\Ldens{c}, \Delta_0)) \leq \xi/2$.

Let us denote $x_{\ell} = \batta(T^{(\ell)}(\lfloor\Ldens{c}\rfloor_\sym{\Kchannel}, \Delta_0))$. Using extremes of information combining \cite{RiU08} we get $x_{\ell} \leq \batta(\lfloor\Ldens{c}\rfloor_\sym{\Kchannel})\lambda(1 - \rho(1 - x_{\ell-1}))$. Expanding around zero, we get $x_{\ell} \leq \batta(\lfloor\Ldens{c}\rfloor_\sym{\Kchannel})\lambda'(0)\rho'(1)x_{\ell-1} + O(x^2_{\ell-1})$. Using the hypothesis of the lemma, lemma~\ref{lem:DistSymClip} and \text{(ix), Lem. 13 in \cite{KRU12b}} we have, $\batta(\lfloor\Ldens{c}\rfloor_\sym{\Kchannel})\lambda'(0)\rho'(1) < 1$. Hence, there exists $\eta > 0$ such that $\batta(\lfloor\Ldens{c}\rfloor_\sym{\Kchannel})\lambda'(0)\rho'(1) + \eta < 1$. From above we know that there exists $\ell$  (and consequently $\Kchannel$ large enough) such that the second order term $O(x^2_{\ell-1})$ is upper bounded by $\eta x_{\ell-1}$. Thus we get $x_{\ell} \leq (\batta(\lfloor\Ldens{c}\rfloor_\sym{\Kchannel})\lambda'(0)\rho'(1)+ \eta)x_{\ell-1} < x_{\ell-1}$. Thus $x_\ell \to 0$ as $\ell \to \infty$ and we get the lemma.

The loss in capacity is bounded by using the Wasserstein distance. Thus $d(\Ldens{c},\lfloor\Ldens{c}\rfloor_\sym{\Kchannel}) \leq 1 - \tanh(\Kchannel/2)$ implies $\entropy(\lfloor\Ldens{c}\rfloor_\sym{\Kchannel}) \leq \entropy(\Ldens{c}) + \frac{2}{\ln 2}e^{-\Kchannel/2}$. Above we have used $1 - \tanh(\Kchannel/2) \leq 2e^{-\Kchannel}$ and \text{(ix), Lem. 13 in \cite{KRU12b}}. Thus,
$1 - \entropy(\lfloor\Ldens{c}\rfloor_\sym{\Kchannel}) \geq 1 - \entropy(\Ldens{c}) - \frac{2}{\ln 2}e^{-\Kchannel/2}$.
\end{proof}

From the above lemma and the analysis in section~\ref{sec:nonsymSatBP} we get\footnote{Recall that we associated a uniform random variable to each variable node which were used for the flipping operations for outgoing messages from the variable node side. For the present case, we can associate a random variable to each channel input which is used for the flipping operation for symmetrizing the saturated channel. These two operations are independent of each other. In  section~\ref{sec:nonsymSatBP} the event $A_\Kvar$ now corresponds to the event that there are no flips at both the variable node and channel input. This probability will be lower bounded by $1 - 2e^{-\Kvar} |V(\sf{T})|$.} $\batta(T^{(\ell)}(\lfloor\Ldens{c}\rfloor_{\Kchannel}, \Delta_0))  \leq \frac{1}{1-\epsilon}\batta(T^{(\ell)}(\lfloor\Ldens{c}\rfloor_\sym{\Kchannel}, \Delta_0))$, for any $0<\epsilon < 1$. Since $\Ldens{c} \prec \SatLdens{c}_\Kchannel \prec \SatLdens{c}_\sym{\Kchannel},$ we have $\entropy(\lfloor\Ldens{c}\rfloor_{\Kchannel}) \leq \entropy(\lfloor\Ldens{c}\rfloor_\sym{\Kchannel})$ which implies that $1 - \entropy(\lfloor\Ldens{c}\rfloor_{\Kchannel}) \geq 1 - \entropy(\Ldens{c}) - \frac{2}{\ln 2}e^{-\Kchannel/2}$.

Note that the stability analysis of section~\ref{sec:stabilityanalysis} does not rely on the symmetry of the channel. The symmetry allows us to  that Battacharyya parameter of the channel is less than one, which is then used to show bounds. In the present case, since 
$\batta(\lfloor\Ldens{c}\rfloor_{\Kchannel}) \leq \batta(\Ldens{c}) + e^{-\Kchannel/2}$  we can proceed with the stability analysis as before and conclude that the SatBP decoder is successful when we first truncate the channel to a large but finite support. Furthermore, this truncation causes minimal loss in the maximum number of information bits that can be transmitted. Finally, we can also say that for any channel $\Ldens{c} \prec \Ldens{c}^{\BPsmall}$ such that $\batta(\Ldens{c}) < \batta(\Ldens{c}^{\BPsmall}) - 2e^{-\Kchannel/2}$, the SatBP decoder is successful over the truncated channel. Thus, the loss in the BP threshold is also upper bounded by $Ce^{-\Kchannel/2}$ for some constant $C$. Note that the threshold for the SatBP decoder is now defined with respect to the fixed point with Battacharyya parameter equal to $e^{-{\Kvar}/2}$.



\section{Conclusions and Outlook}
In this paper we perform perturbation analysis of the standard LDPC code ensemble and BP decoder combination. 
Specifically, we show that saturating the messages arising in the BP decoding process affects the final success of the decoder.
For general irregular LDPC code ensembles with minimum variable node degree three, we show that the saturation of the messages still allows for successful decoding as long as the saturation level $\Kvar$ is large enough. More precisely, whenever the channel is below the BP threshold, then there exists a saturation value $\Kvar$, which is large enough but finite, such that the SatBP decoder is also below its threshold.
The stability of the SatBP decoder requires the support of the channel to be finite. In the case of channels with infinite support, we show that by saturating the channel first to a large enough value, we sacrifice little in terms of capacity. Then, on the saturated channel, the SatBP decoder is successful. Thus there is minimal sacrifice in the BP threshold of the LDPC code ensemble when we consider the SatBP decoder. 

When the minimum variable node degree is two the saturated decoding system fails to have stability of perfect decoding. We show that the perfect decoding fixed point (the delta function at $\Kvar$) cannot be a stable fixed point of DE for the SatBP decoder unless the channel is the erasure channel.
The key issue is that a density update at a degree two node variable nodes is convolution with the channel density.
Repeated $k$ times, this involves to convolution of the channel density with itself $k$ times.
In general this is equivalent to a channel density with support width $k$ times wider than the original channel.
If the incoming density is saturated then for $k$ large enough a positive error probability is unavoidable.
If the code structure (e.g. protograph designs) ensures that the number of successive degree two node updates in 
the density evolution is bounded, then the expansion $k$ is bounded and one can again recover stability with large enough
saturation.  Essentially, what is required is that each degree two variable node subgraph connected component (asymptotically a tree)
have bounded size.

To give a more detailed indication of how this can work we consider the min-sum decoder and show that
perfect decoding can be invariant even in the presence of degree two variable nodes.
Let the maximum component size be denoted by $A$. 
For an edge $e$ connected to a degree two variable node let $2 L_e + 1$ denote the maximum path length to the edge of
the connected component.  Note that $L_e + 1 \le A.$
To show invariance of a perfect decoding we assume $2 (\Kvar - A \Kchannel ) - \Kchannel \ge \Kvar.$
Assume in some iteration that the following hold,
\begin{itemize}
\item
The incoming message to a degree two variable node with edges $e_1,e_2$ on edge $e_i$ is at least $\Kvar - L_{e_i} \Kchannel.$
\item
Incoming messages on a degree three or higher variable node are at least $\Kvar - A \Kchannel.$
\end{itemize}
It is easy to check that this implies perfect decoding.
Proceeding to the next iteration we obtain,
\begin{itemize}
\item
The outgoing message on a degree two variable node on edge $e_2$ is at least $\Kvar - (L_{e_1}+1) \Kchannel$ (and vice-versa for $e_1.$)
\item
Outgoing messages on a degree three or higher variable node are at least $\Kvar.$
\end{itemize}
Now consider the subsequent incoming messages to the variable nodes.
The minimum outgoing message from the previous iteration is at least $\Kvar-A\Kchannel$ so incoming messages to 
a degree three or higher variable node are at least $\Kvar - A \Kchannel.$
Consider edge $e_1$ attached to a degree two variable node.
The longest path, not traversing $e_1,$ from its neighboring check node to a leaf check of the degree two connect component
has edge length at most $2 L_{e_1}.$  Hence the minimum incoming message to the neighbor check node not from $e_1$ is
 $\Kvar - L_{e_1} \Kchannel.$  The minimum incoming message on edge $e_1$ to the degree two variable node is therefore at least
$\Kvar - L_{e_1} \Kchannel.$  Thus, under the stated assumptions the above perfect decoding conditions are invariant.

\subsection*{Future Directions:} To complete the story of the analysis of the BP decoder under practical considerations, it would be nice to have the analysis of the quantized BP decoder. Thus, the messages are only allowed to take certain values on the real line. Every message is quantized to a bin and only the bin value is passed around. 
For the ease of analysis one can assume a uniformly quantized message space. It is not hard to see that such a quantized BP decoder is symmetric. Thus the standard DE analysis is applicable to the quantized BP decoder. A clear next step would be to see if the analysis performed for the SatBP decoder goes through for the quantized BP decoder. 
If yes, then it would be nice to see a unified perturbation analysis of saturated  and quantized messages. 

A nice side-effect of the analysis done above is that when there are degree three variable nodes present in the LDPC code, it is perhaps better to erase the channel information at those bits completely (after enough iterations are performed) to allow faster convergence to the correct codeword. This sheds some light on the practical design of BP decoders under saturation of messages. Could we glean similar lessons for practical decoder design when we consider the saturated and quantized BP decoder?

Another research direction would be to quantify the saturation and quantization levels in terms of gap to capacity. Specifically, what should be the scaling of the saturation and quantization value when we backoff, say, $\delta$ from the BP capacity, $\ent^{\BPsmall}$. It seems intuitive that as we backoff more  from $\ent^{\BPsmall}$ we should be able to attain the same error rate with smaller values of the saturation level and larger levels of quantization. In other words, as the gap to capacity increases, we should require lesser number of bits in the binary representation of the messages to get the desired error rate.

\begin{appendices}

\section{Battacharrya Parameter Inequality -- Lemma~\ref{lem:additivecheckbound}}\label{app:checknodebatta}

We require the following inequality
\begin{lemma}\label{lem:pprodineq}
Let $p_1,...,p_k$ each lie in $[0,1].$ Then
\[
\frac{1-\prod_{i=1}^k (1-2p_i)}{2} \le \sum_{i=1}^k p_i
\]
\end{lemma}
\begin{IEEEproof}
We have equality when $p_i=0$ for each $i.$
Differentiating the left hand side with respect to $p_j$ we obtain
$\prod_{\{ i\in [1:k]\backslash j\}}  (1-2p_i)$ which has magnitude at most $1$ and 
differentiating the left hand side with respect to $p_j$ we obtain $1.$
The inequality therefore follows by integration.
\end{IEEEproof}

The following generalizes Lemma~\ref{lem:additivecheckbound}.
\begin{lemma}\label{lem:checkbattabound}
Let $D_1,D_2,...D_k$ be L-densities of the form $D_i = D(p_i,\Kgen)$
and let $\Ldens{a}_1,\ldots,\Ldens{a}_{d-k}$ be L-densities.
We do not assume that any of these densities are symmetric.
Let $\Ldens{b}$ denote the density emerging from a check node update
when the incoming densities are $D_1,...,D_k,\Ldens{a}_1,\ldots,\Ldens{a}_{d-k},$ then
\[
\batta{(\Ldens{b})} 
\le
\bigl(1+\sum_{i=1}^k (e^{{\Kgen}/2} \batta{(D_i)}-1)\bigr)
\bigl(\sum_{i=j}^{d-k}  \batta{(\Ldens{a}_j)}\bigr)\,.
\]
\end{lemma}
(This holds even if $k=0$ in which case we have only the second factor.)
This generalizes a result from \cite{BRU07}.
\begin{IEEEproof}
By averaging, we see that it is sufficient to prove the lemma
for the case $\Ldens{a}_i = D(q_i,z_i).$ 
With this assumption the outgoing message is of the form
$\Ldens{b} = D(s,r)$ where
\[
s = \frac{1-(\prod_{i=1}^k (1-2p_i))(\prod_{j=1}^{d-k} (1-2q_j))}{2}\,, 
\]
and we have $r \le \min \{ \Kgen, q_1,...,q_{d-k} \}$ and
$e^{-r/2} \le k e^{-{\Kgen}/2}+\sum_{j=1}^{d-k}e^{-q_i/2}.$
We have $\batta{(\Ldens{b})} =s e^{r/2}+(1-s)e^{-r/2}.$

Define
\begin{align*}
P= \frac{1-\prod_{i=1}^k (1-2p_i)}{2},\,\,\,\,
Q=\frac{1-\prod_{j=1}^{d-k} (1-2q_j)}{2} 
\end{align*}
Then we have
\[
1-s = PQ+(1-P)(1-Q)\,.
\]
We claim the inequality
\[
\batta{(\Ldens{b})} \le
(Pe^{\Kgen}+(1-P))(Qe^{r/2}+(1-Q)e^{-r/2})\,.
\]
The claim follows from collecting terms and noting
$e^{\Kgen} e^{r/2} \ge e^{-r/2},$ which is obvious,
and
$e^{\Kgen} e^{-r/2} \ge e^{r/2},$ which follows from $\Kvar \ge r.$

We now apply Lemma \ref{lem:pprodineq} to the left factor to obtain
\begin{align*}
Pe^{\Kgen}&+(1-P)  = 1+P(e^{\Kgen}-1) \\
&\le
1+(\sum_{i=1}^k p_i)(e^{\Kgen}-1) \\
&=
1+\sum_{i=1}^k \bigl( e^{{\Kgen}/2}(p_i e^{{\Kgen}/2}+(1-p_i) e^{-{\Kgen}/2}) -1\bigl) \\
&=
1+\sum_{i=1}^k (e^{{\Kgen}/2} \batta{(D_i)}-1)\,.
\end{align*}

Using $q_j \le r$ and $\sum_{j=1}^{d-k} e^{-q_j/2} \ge e^{-r/2}$
and applying Lemma \ref{lem:pprodineq} to the right factor we obtain
\begin{align*}
Qe^{r/2}+(1-Q)e^{-r/2} & = e^{-r/2}+Q(2\sinh (r/2)) \\
&\le
 e^{-r/2}+(\sum_{j=1}^{d-k} q_j)(2\sinh (r/2)) \\
&\le
\sum_{j=1}^{d-k} e^{-q_j/2}
+\sum_{j=1}^{d-k} q_j (2\sinh (q_j/2)) \\
&=
\sum_{i=1}^{d-k}  \batta{(\Ldens{a}_j)}\,.
\end{align*}
\end{IEEEproof}

\end{appendices}

\bibliographystyle{IEEEtran}
\bibliography{lth,lthpub,extras}

\newcommand{\SortNoop}[1]{}
\begin{thebibliography}{10}
\providecommand{\url}[1]{#1}
\csname url@rmstyle\endcsname
\providecommand{\newblock}{\relax}
\providecommand{\bibinfo}[2]{#2}
\providecommand\BIBentrySTDinterwordspacing{\spaceskip=0pt\relax}
\providecommand\BIBentryALTinterwordstretchfactor{4}
\providecommand\BIBentryALTinterwordspacing{\spaceskip=\fontdimen2\font plus
\BIBentryALTinterwordstretchfactor\fontdimen3\font minus
  \fontdimen4\font\relax}
\providecommand\BIBforeignlanguage[2]{{%
\expandafter\ifx\csname l@#1\endcsname\relax
\typeout{** WARNING: IEEEtran.bst: No hyphenation pattern has been}%
\typeout{** loaded for the language `#1'. Using the pattern for}%
\typeout{** the default language instead.}%
\else
\language=\csname l@#1\endcsname
\fi
#2}}

\bibitem{RiU08}
T.~Richardson and R.~Urbanke, \emph{Modern Coding Theory}.\hskip 1em plus 0.5em
  minus 0.4em\relax Cambridge University Press, 2008.

\bibitem{6290251}
X.~Zhang and P.~Siegel, ``Will the real error floor please stand up?'' in
  \emph{Signal Processing and Communications (SPCOM), 2012 International
  Conference on}, 2012, pp. 1--5.

\bibitem{6120169}
B.~Butler and P.~Siegel, ``Error floor approximation for ldpc codes in the awgn
  channel,'' in \emph{Communication, Control, and Computing (Allerton), 2011
  49th Annual Allerton Conference on}, 2011, pp. 204--211.

\bibitem{5485006}
C.~Schlegel and S.~Zhang, ``On the dynamics of the error floor behavior in
  (regular) ldpc codes,'' \emph{Information Theory, IEEE Transactions on},
  vol.~56, no.~7, pp. 3248--3264, 2010.

\bibitem{6567866}
S.~Zhang and C.~Schlegel, ``Controlling the error floor in ldpc decoding,''
  \emph{Communications, IEEE Transactions on}, vol.~61, no.~9, pp. 3566--3575,
  2013.

\bibitem{6284049}
X.~Zhang and P.~Siegel, ``Quantized min-sum decoders with low error floor for
  ldpc codes,'' in \emph{Information Theory Proceedings (ISIT), 2012 IEEE
  International Symposium on}, 2012, pp. 2871--2875.

\bibitem{6685976}
------, ``Quantized iterative message passing decoders with low error floor for
  ldpc codes,'' pp. 1--14, 2013.

\bibitem{VNC14}
B.~Vasic, D.~V. Nguyen, and S.~K. Chilappagari, ``Failures and error-floors of
  iterative decoders,'' \emph{Channel Coding: Theory, Algorithms, and
  Applications, Academic Press Library in Mobile and Wireless, Communications,
  Elsevier, New York}, 2014.

\bibitem{6134417}
J.~Wang, T.~Courtade, H.~Shankar, and R.~Wesel, ``Soft information for ldpc
  decoding in flash: Mutual-information optimized quantization,'' in
  \emph{Global Telecommunications Conference (GLOBECOM 2011), 2011 IEEE}, 2011,
  pp. 1--6.

\bibitem{5624882}
N.~Kanistras, I.~Tsatsaragkos, I.~Paraskevakos, A.~Mahdi, and V.~Paliouras,
  ``Impact of llr saturation and quantization on ldpc min-sum decoders,'' in
  \emph{Signal Processing Systems (SIPS), 2010 IEEE Workshop on}, 2010, pp.
  410--415.

\bibitem{6363268}
N.~Kanistras, I.~Tsatsaragkos, and V.~Paliouras, ``Propagation of llr
  saturation and quantization error in ldpc min-sum iterative decoding,'' in
  \emph{Signal Processing Systems (SiPS), 2012 IEEE Workshop on}, 2012, pp.
  276--281.

\bibitem{5205826}
Y.~Wu, L.~Davis, and R.~Calderbank, ``On the capacity of the discrete-time
  channel with uniform output quantization,'' in \emph{Information Theory,
  2009. ISIT 2009. IEEE International Symposium on}, 2009, pp. 2194--2198.

\bibitem{KRU12}
S.~Kudekar, T.~Richardson, and R.~L. Urbanke, ``Wave-like solutions of general
  one-dimensional spatially coupled systems,'' \emph{CoRR}, vol. abs/1208.5273,
  2012.

\bibitem{RiU01}
T.~Richardson and R.~Urbanke, ``The capacity of low-density parity check codes
  under message-passing decoding,'' \emph{IEEE Trans. Inform. Theory}, vol.~47,
  no.~2, pp. 599--618, Feb. 2001.

\bibitem{HaL69}
G.~Hanoch and H.~Levy, ``The efficiency analysis of choices involving risk,''
  \emph{The Review of Economic Studies}, vol.~36, pp. 335--346, 1969.

\bibitem{KRU11Wasserstein}
S.~Kudekar, T.~Richardson, and R.~Urbanke, ``Existence and {U}niqueness of
  {GEXIT} curves via the {W}asserstein {M}etric,'' in \emph{Proc. of the IEEE
  Inform. Theory Workshop}, Paraty, Brazil, 2011.

\bibitem{Villani09}
C.~Villani, \emph{Optimal transport, Old and New}.\hskip 1em plus 0.5em minus
  0.4em\relax Springer, 2009, vol. 338.

\bibitem{KRU12b}
S.~Kudekar, T.~Richardson, and R.~L. Urbanke, ``Spatially coupled ensembles
  universally achieve capacity under belief propagation,'' \emph{CoRR}, vol.
  abs/1201.2999, 2012.

\bibitem{1523291}
H.~Jin and T.~Richardson, ``Block error iterative decoding capacity for ldpc
  codes,'' in \emph{Information Theory, 2005. ISIT 2005. Proceedings.
  International Symposium on}, Sept 2005, pp. 52--56.

\bibitem{1522644}
M.~Lentmaier, D.~Truhachev, K.~Zigangirov, and D.~Costello, ``An analysis of
  the block error probability performance of iterative decoding,''
  \emph{Information Theory, IEEE Transactions on}, vol.~51, no.~11, pp.
  3834--3855, Nov 2005.

\bibitem{BRU07}
K.~Bhattad, V.~Rathi, and R.~Urbanke, ``Degree optimization and stability
  condition for the min-sum decoderl,'' in \emph{Proc. of the IEEE Inform.
  Theory Workshop}, 2007, conference, pp. 190--195.

\end{thebibliography}

\end{document}